\documentclass[10pt,draft,journal,a4paper,oneside,onecolumn]{IEEEtran} 
\usepackage{amsmath,amsfonts, amssymb}

\newtheorem{thm}{Theorem}
\newtheorem{lem}{Lemma}
\newtheorem{df}{Definition}
\newtheorem{rem}{Remark}

\newcommand{\disp}{\displaystyle}

\newcommand{\GFq}{\mathrm{GF}(q)}
\newcommand{\Real}{\mathbb{R}}
\newcommand{\hA}{\widehat{A}}
\newcommand{\hB}{\widehat{B}}
\newcommand{\hg}{\widehat{g}}
\newcommand{\A}{\mathcal{A}}
\newcommand{\hcA}{\widehat{\mathcal{A}}}
\newcommand{\bcA}{\boldsymbol{\mathcal{A}}}
\newcommand{\bhcA}{\boldsymbol{\widehat{\mathcal{A}}}}
\newcommand{\B}{\mathcal{B}}
\newcommand{\hcB}{\widehat{\mathcal{B}}}
\newcommand{\bhcB}{\boldsymbol{\widehat{\mathcal{B}}}}
\newcommand{\bcB}{\boldsymbol{\mathcal{B}}}
\newcommand{\C}{\mathcal{C}}
\newcommand{\E}{\mathcal{E}}
\newcommand{\G}{\mathcal{G}}
\newcommand{\cH}{\mathcal{H}}
\newcommand{\hcH}{\widehat{\mathcal H}}
\newcommand{\J}{\mathcal{J}}
\newcommand{\K}{\mathcal{K}}
\newcommand{\cP}{\mathcal{P}}
\newcommand{\cS}{\mathcal{S}}
\newcommand{\T}{\mathcal{T}}
\newcommand{\U}{\mathcal{U}}
\newcommand{\bU}{\overline{\mathcal{U}}}
\newcommand{\V}{\mathcal{V}}

\newcommand{\X}{\mathcal{X}}
\newcommand{\Y}{\mathcal{Y}}
\newcommand{\Z}{\mathcal{Z}}
\newcommand{\aalpha}{\boldsymbol{\alpha}}
\newcommand{\bbeta}{\boldsymbol{\beta}}
\newcommand{\kkappa}{\boldsymbol{\kappa}}
\newcommand{\haa}{\boldsymbol{\widehat a}}
\newcommand{\hbb}{\boldsymbol{\widehat b}}
\newcommand{\sfhaa}{\boldsymbol{\widehat{\mathsf a}}}
\newcommand{\sfhbb}{\boldsymbol{\widehat{\mathsf b}}}
\newcommand{\ba}{\boldsymbol{a}}
\newcommand{\bb}{\boldsymbol{b}}
\newcommand{\cc}{\boldsymbol{c}}
\newcommand{\ff}{\boldsymbol{f}}
\newcommand{\FF}{\boldsymbol{F}}
\newcommand{\mm}{\boldsymbol{m}}
\newcommand{\bp}{\boldsymbol{p}}
\newcommand{\bs}{\boldsymbol{s}}
\newcommand{\bt}{\boldsymbol{t}}
\newcommand{\uu}{\boldsymbol{u}}
\newcommand{\UU}{\boldsymbol{U}}
\newcommand{\vv}{\boldsymbol{v}}
\newcommand{\VV}{\boldsymbol{V}}
\newcommand{\ww}{\boldsymbol{w}}
\newcommand{\xx}{\boldsymbol{x}}
\newcommand{\XX}{\boldsymbol{X}}
\newcommand{\yy}{\boldsymbol{y}}
\newcommand{\zz}{\boldsymbol{z}}
\newcommand{\e}{\varepsilon}
\newcommand{\sfA}{\mathsf{A}}
\newcommand{\sfB}{\mathsf{B}}
\newcommand{\sfhA}{\widehat{\mathsf{A}}}
\newcommand{\sfhB}{\widehat{\mathsf{B}}}
\newcommand{\sfaa}{\boldsymbol{\mathsf{a}}}
\newcommand{\sfbb}{\boldsymbol{\mathsf{b}}}
\newcommand{\sfmm}{\boldsymbol{\mathsf{m}}}
\newcommand{\sfww}{\boldsymbol{\mathsf{w}}}
\newcommand{\Prod}{\operatornamewithlimits{\text{\Large $\times$}}}
\newcommand{\lrB}[1]{\left[{#1}\right]}
\newcommand{\lrb}[1]{\left\{{#1}\right\}}
\newcommand{\lrsb}[1]{\left({#1}\right)}
\newcommand{\lrbar}[1]{\left|{#1}\right|}
\newcommand{\Error}{\mathrm{Error}}
\newcommand{\zero}{\boldsymbol{0}}
\newcommand{\one}{\boldsymbol{1}}
\newcommand{\limn}{\lim_{n\to\infty}}
\newcommand{\Encoder}{\varphi}
\newcommand{\Decoder}{\varphi^{-1}}
\newcommand{\Rate}{\mathrm{Rate}}
\newcommand{\im}{\mathrm{Im}}

\newcommand{\bpAB}{\bp_{\sfA\sfB}}
\newcommand{\bpA}{\bp_{\sfA}}
\newcommand{\bpB}{\bp_{\sfB}}
\newcommand{\bpAp}{\bp_{\sfA'}}
\newcommand{\bpBp}{\bp_{\sfB'}}
\newcommand{\bphA}{\bp_{\sfhA}}
\newcommand{\bphB}{\bp_{\sfhB}}
\newcommand{\pA}{p_{\sfA}}
\newcommand{\pB}{p_{\sfB}}
\newcommand{\pa}{p_{\sfaa}}

\newcommand{\pAa}{p_{\sfA\sfaa}}
\newcommand{\pAp}{p_{\sfA'}}
\newcommand{\phA}{p_{\sfhA}}
\newcommand{\alphaA}{\alpha_{\sfA}}
\newcommand{\alphaAp}{\alpha_{\sfA'}}
\newcommand{\alphaB}{\alpha_{\sfB}}
\newcommand{\alphaBp}{\alpha_{\sfB'}}
\newcommand{\betaA}{\beta_{\sfA}}
\newcommand{\betaAp}{\beta_{\sfA'}}
\newcommand{\betaB}{\beta_{\sfB}}
\newcommand{\betaBp}{\beta_{\sfB'}}
\newcommand{\aalphaA}{\aalpha_{\sfA}}
\newcommand{\aalphaAp}{\aalpha_{\sfA'}}
\newcommand{\bbetaA}{\bbeta_{\sfA}}
\newcommand{\bbetaAp}{\bbeta_{\sfA'}}
\newcommand{\bbetaB}{\bbeta_{\sfB}}

\newcommand{\alphahA}{\alpha_{\sfhA}}
\newcommand{\alphahB}{\alpha_{\sfhB}}
\newcommand{\betahA}{\beta_{\sfhA}}
\newcommand{\betahB}{\beta_{\sfhB}}
\newcommand{\aalphahA}{\aalpha_{\sfhA}}
\newcommand{\bbetahA}{\bbeta_{\sfhA}}
\newcommand{\aalphahB}{\aalpha_{\sfhB}}
\newcommand{\bbetahB}{\bbeta_{\sfhB}}

\allowdisplaybreaks

\title{
  Construction of
  Slepian-Wolf Source Code and
  Broadcast Channel Code\\
  Based on Hash Property
}
\author{
  Jun~Muramatsu and~Shigeki Miyake
  \thanks{
    J.~Muramatsu is with
    NTT Communication Science Laboratories, NTT Corporation,
    2-4, Hikaridai, Seika-cho, Soraku-gun, Kyoto 619-0237, Japan
    (E-mail: muramatsu.jun@lab.ntt.co.jp).
    S.~Miyake is with
    NTT Network Innovation Laboratories, NTT Corporation,
    Hikarinooka 1-1, Yokosuka-shi, Kanagawa 239-0847, Japan
    (E-mail: miyake.shigeki@lab.ntt.co.jp).
  }
  \thanks{A part of this paper was presented in part at
    {\em Proc. IEEE Int. Symp. on Inform. Theory} (ISIT2010)
    and
    {\em Proc. 7-th-Asia-Europe Workshop ``CONCEPTS in INFORMATION
      THEORY''} (AEW7), 2011.
    Proof of Lemma 9 was revised in May 22, 2012.
    Proof of Lemma 4 was revised in Jan. 25, 2013.
  }
}

\date{June 27, 2010, revised Jan. 25, 2013}
\begin{document}
\maketitle

\begin{abstract}
  The aim of this paper is to prove theorems for the Slepian-Wolf source
  coding and the broadcast channel coding (independent messages and no
    common message) based on the the notion of a stronger version of the
  hash property for an ensemble of functions.
  Since an ensemble of sparse matrices (with logarithmic column degree)
  has a strong hash property,
  codes using sparse matrices can realize the achievable rate region.
  Furthermore, extensions to the multiple source coding and
  multiple output broadcast channel coding are investigated.
\end{abstract}

\section{Introduction}
The aim of this paper is to prove theorems for the Slepian-Wolf source
coding (Fig.~\ref{fig:sw}) introduced in \cite{SW73} and the broadcast
channel coding (Fig.~\ref{fig:broadcast}) introduced in \cite{C72}.
The proof of theorems is based on a stronger version of the hash
property for an ensemble of functions introduced
in~\cite{HASH}\cite{HASH-UNIV}.
This notion provides a sufficient condition for the achievability
of coding theorems.
Since an ensemble of sparse matrices also has a strong hash property,
we can construct codes by using sparse matrices where the rate pair of
the code is close to the boundary of the achievable rate region.
When constructing codes,
we employ minimum-divergence encoding and maximum-likelihood decoding.

The achievable rate region for Slepian-Wolf source coding is derived
in~\cite{SW73}.
The technique of random bin coding is developed in \cite{C75} to prove
the Slepian-Wolf source coding theorem.
The achievability of the code using a pair of matrices
is studied in \cite{CSI82}.

The constructions of encoders using sparse matrices is studied
in~\cite{M02}\cite{SPR02} and the achievability is proved
in~\cite{SWLDPC} by using maximum-likelihood (ML) decoding.
In this paper, we construct codes based on the strong hash property,
which unifies the results of \cite{C75}\cite{CSI82}\cite{SWLDPC}.

To construct a broadcast channel code,
we assume that independent messages are decoded by their respective
receivers with small error probability.
It should be noted that we assume neither ``degraded'' nor ``less
noisy'' conditions on this channel.
The capacity region for this channel is known only for some classes.
An inner region on the two receiver general broadcast channel without a
common message is derived in \cite{M79} and a simpler way of
constructing codes is presented in \cite{EM81}.
The cardinality bound is investigated in \cite{GA09}\cite{HP79}.
Outer regions are derived in \cite{LKS08}\cite{N08}\cite{NE07}.
Applications of sparse matrices (LDPC codes) to broadcast channels are
investigated in \cite{BT05}\cite{NKMS03}\cite{O07}\cite{RA09}.
In this paper, we construct codes based on a strong hash property
and we show that the rate pair of the constructed code is close to the
inner bound derived in \cite{EM81}\cite{M79}.

\begin{figure}
  \begin{center}
    \unitlength 0.5mm
    \begin{picture}(155,45)(0,0)
      \put(5,25){\makebox(0,0){$X$}}
      \put(10,25){\vector(1,0){10}}
      \put(27,37){\makebox(0,0){Senders}}
      \put(20,18){\framebox(14,14){$\Encoder_{X}$}}
      \put(34,25){\vector(1,0){84}}
      \put(5,7){\makebox(0,0){$Y$}}
      \put(10,7){\vector(1,0){10}}
      \put(20,0){\framebox(14,14){$\Encoder_{Y}$}}
      \put(34,7){\vector(1,0){84}}
      \put(125,37){\makebox(0,0){Receiver}}
      \put(118,0){\framebox(14,32){$\Decoder$}}
      \put(132,16){\vector(1,0){10}}
      \put(152,16){\makebox(0,0){$(X,Y)$}}
    \end{picture}
  \end{center}
  \caption{Slepian-Wolf Source Coding}
  \label{fig:sw}
\end{figure}
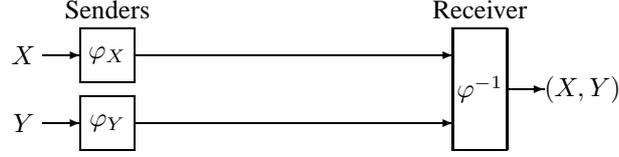

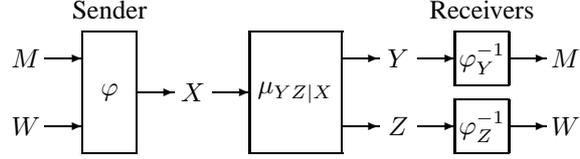
\begin{figure}
  \begin{center}
    \unitlength 0.5mm
    \begin{picture}(155,45)(0,10)
      \put(5,35){\makebox(0,0){$M$}}
      \put(10,35){\vector(1,0){10}}
      \put(5,17){\makebox(0,0){$W$}}
      \put(10,17){\vector(1,0){10}}
      \put(27,48){\makebox(0,0){Sender}}
      \put(20,10){\framebox(14,32){$\Encoder$}}
      \put(34,26){\vector(1,0){10}}
      \put(49,26){\makebox(0,0){$X$}}
      \put(54,26){\vector(1,0){10}}
      \put(64,10){\framebox(24,32){$\mu_{YZ|X}$}}
      \put(88,35){\vector(1,0){10}}
      \put(88,17){\vector(1,0){10}}
      \put(103,35){\makebox(0,0){$Y$}}
      \put(103,17){\makebox(0,0){$Z$}}
      \put(108,35){\vector(1,0){10}}
      \put(108,17){\vector(1,0){10}}
      \put(125,48){\makebox(0,0){Receivers}}
      \put(118,28){\framebox(14,14){$\Decoder_Y$}}
      \put(132,35){\vector(1,0){10}}
      \put(147,35){\makebox(0,0){$M$}}
      \put(118,10){\framebox(14,14){$\Decoder_Z$}}
      \put(132,17){\vector(1,0){10}}
      \put(147,17){\makebox(0,0){$W$}}
    \end{picture}
  \end{center}
  \caption{Broadcast Channel Coding}
  \label{fig:broadcast}
\end{figure}

The proof of all theorems is based on the notion of a hash property,
which is a stronger version of that introduced in \cite{HASH}.
It is the extension of the ensemble of the random bin coding~\cite{C75},
the ensembles of linear matrices~\cite{CSI82},
the universal class of hash functions~\cite{CW},
and the ensemble of sparse matrices~\cite{SWLDPC}.
Two lemmas called `collision-resistace property\footnote{In~\cite{HASH},
  they were called  `collision-resistant property' and
  `saturating property,' respectively.
  We changed these terms following the suggestion of Prof. T.S.~Han.}'
(if the number of bins is greater than the number of items then there is
  an assignment such that every bin contains at most one item)
and `saturation property$^{\text{1}}$'
(if the number of items is greater than the number of bins then there is
  an assignment such that every bin contains at least one item),
which are proved in \cite{HASH} and reviewed in Section \ref{sec:hash},
are extended from a single domain to multiple domains.
The extended collision-resistance property is used to analyze the
decoding error of the Slepian-Wolf source coding and the extended
saturation property is used to analyze the encoding error of the
broadcast channel coding.
It should be noted that the linearity of functions is not necessary
for a strong hash property but it is expected that the space and time
complexity of the code can be reduced compared with conventional
constructions by using sparse matrices.
This is a potential advantage of our approach.

\section{Definitions and Notations}
Throughout this paper, we use the following definitions and notations.

The cardinality of a set $\U$ is denoted by $|\U|$,
$\U^c$ denotes the complement of $\U$,
and $\U\setminus\V\equiv\U\cap\V^c$ denotes the set difference.
Column vectors and sequences are denoted in boldface.
For $\T\subset\U^n\times\V^n$ and $\vv\in\V^n$,
$\T_{\U}$ and $\T_{\U|\V}(\vv)$ are defined as
\begin{align*}
  \T_{\U}&\equiv\{\uu: (\uu,\vv)\in\T\ \text{for some}\ \vv\in\V^n\}
  \\
  \T_{\U|\V}(\vv)&\equiv\{\uu: (\uu,\vv)\in\T\}.
\end{align*}

Let $A\uu$ denote a value taken by a function $A:\U^n\to\im A$
at $\uu\equiv(u_1,\ldots,u_n)\in\U^n$,
where $\U^n$ and $\im A\equiv\{A\uu: \uu\in\U^n\}$ are the domain and
the image of the function, respectively.
It should be noted that $A$ may be nonlinear.
When $A$ is a linear function expressed by an $l\times n$ matrix,
we assume that $\U\equiv\GFq$ is a finite field.
For a set $\A$ of functions, let $\im \A$ be defined as
\begin{align*}
  \im\A &\equiv \bigcup_{A\in\A}\im A.
\end{align*}
We define a set $\C_A(\ba)$ as
\begin{align*}
  \C_A(\ba)
  &\equiv\{\uu: A\uu = \ba\}.
\end{align*}
In the context of linear codes,
$\C_A(\ba)$ is called a coset determined by $\ba$.
The random variables of a function $A$ and a vector $\ba\in\im\A$ are
denoted by sans serif letters $\sfA$ and $\sfaa$, respectively.
On the other hand, the random variables of a $n$-dimensional vector
$\uu$ is denoted by bold Roman letter $\UU$.

Let $p$ and $p'$ be probability distributions and let $q$ and $q'$ be
conditional probability distributions.
Then entropy $H(p)$, conditional entropy $H(q|p)$,
divergence $D(p\|p')$, and conditional divergence $D(q\|q'|p)$
are defined as
\begin{align*}
  H(p)
  &\equiv\sum_{u}p(u)\log\frac 1{p(u)}
  \\
  H(q|p)
  &\equiv\sum_{u,v}q(u|v)p(v)\log\frac 1{q(u|v)}
  \\
  D(p\parallel p')
  &\equiv
  \sum_{u}p(u)
  \log\frac{p(u)}{p'(u)}
  \\
  D(q\parallel q' | p)
  &\equiv
  \sum_{v} p(v)\sum_{u}q(u|v)
  \log\frac{q(u|v)}{q'(u|v)},
\end{align*}
where we assume that the base of the logarithm is $2$.

Let $\mu_{UV}$ be the joint probability distribution of random variables
$U$ and $V$.
Let  $\mu_{U}$ and $\mu_{V}$ be the respective marginal distributions
and $\mu_{U|V}$ be the conditional probability distribution.
Then the entropy $H(U)$, the conditional entropy $H(U|V)$, and the mutual
information $I(U;V)$ of random variables are defined as
\begin{align*}
  H(U)&\equiv H(\mu_U)
  \\
  H(U|V)&\equiv H(\mu_{U|V}|\mu_{V})
  \\
  I(U;V)&\equiv
  H(U)-H(U|V).
\end{align*}

A set of typical sequences $\T_{U,\gamma}$ and a set of conditionally
typical sequences $\T_{U|V,\gamma}(\vv)$ are defined as
\begin{align*}
  \T_{U,\gamma}
  &\equiv
  \lrb{\uu:
    D(\nu_{\uu}\|\mu_{U})<\gamma
  }
  \\
  \T_{U|V,\gamma}(\vv)
  &\equiv
  \lrb{\uu:
    D(\nu_{\uu|\vv}\|\mu_{U|V}|\nu_{\vv})<\gamma
  },
\end{align*}
respectively, where $\nu_{\uu}$ and $\nu_{\uu|\vv}$ are defined as
\begin{align*}
  &\nu_{\uu}(u)
  \equiv
  \frac {|\{1\leq i\leq n : u_{i}=u\}|}n
  \\
  &\nu_{\uu|\vv}(u|v)
  \equiv \frac{\nu_{\uu\vv}(u,v)}{\nu_{\vv}(v)}.
\end{align*}

We define $\chi(\cdot)$ as
\begin{align*}
  \chi(a = b)
  &\equiv
  \begin{cases}
    1,&\text{if}\ a = b
    \\
    0,&\text{if}\ a\neq b
  \end{cases}
  \\
  \chi(a \neq b)
  &\equiv
  \begin{cases}
    1,&\text{if}\ a \neq b
    \\
    0,&\text{if}\ a = b.
  \end{cases}
\end{align*}

Finally, for $\gamma,\gamma'>0$, we define
\begin{align}
  \lambda_{\U}
  &\equiv \frac{|\U|\log(n+1)}n
  \label{eq:lambda}
  \\
  \zeta_{\U}(\gamma)
  &\equiv
  \gamma-\sqrt{2\gamma}\log\frac{\sqrt{2\gamma}}{|\U|}
  \label{eq:zeta}
  \\
  \zeta_{\U|\V}(\gamma'|\gamma)
  &\equiv
  \gamma'-\sqrt{2\gamma'}\log\frac{\sqrt{2\gamma'}}{|\U||\V|}
  +\sqrt{2\gamma}\log|\U|
  \label{eq:zetac}
  \\
  \eta_{\U}(\gamma)
  &\equiv
  -\sqrt{2\gamma}\log\frac{\sqrt{2\gamma}}{|\U|}
  +\frac{|\U|\log(n+1)}n
  \label{eq:def-eta}
  \\
  \eta_{\U|\V}(\gamma'|\gamma)
  &\equiv
  -\sqrt{2\gamma'}\log\frac{\sqrt{2\gamma'}}{|\U||\V|}
  +\sqrt{2\gamma}\log|\U|
  +\frac{|\U||\V|\log(n+1)}n,
  \label{eq:def-etac}
\end{align}
where the product set $\U\times\V$ is denoted by $\U\V$
when it appears in the subscript of these functions.
These definitions will be used in the proof of theorems.

\section{Strong $(\aalpha,\bbeta)$-hash property}
\label{sec:hash}

\subsection{Formal Definition and Basic Properties}
In the following, we introduce the strong hash property for an
ensemble of functions.
It requires a stronger condition than that introduced in \cite{HASH}.
It should be noted that the linearity of functions is not assumed
in this section.

\begin{df}
Let $\bcA\equiv\{\A_n\}_{n=1}^{\infty}$
be a sequence of sets such that
$\A_n$ is a set of functions
$A_n:\U^n\to\im\A_n$.
For a probability distribution $p_{\sfA_n}$ on $\A_n$
corresponding to a random variable $\sfA_n\in\A_n$,
we call a sequence $(\bcA,\bpA)\equiv\{(\A_n,p_{\sfA_n})\}_{n=1}^{\infty}$
an {\em ensemble}.
Then, $(\bcA,\bpA)$ has a {\em strong}
$(\aalphaA,\bbetaA)$-{\em hash property} if
there are two sequences
$\aalphaA\equiv\{\alphaA(n)\}_{n=1}^{\infty}$ and
$\bbetaA\equiv\{\betaA(n)\}_{n=1}^{\infty}$
depending only on $\{p_{\sfA_n}\}_{n=0}^{\infty}$ such that
\begin{align}
  &\limn \alphaA(n)=1
  \tag{H1}
  \label{eq:alpha}
  \\
  &\limn \betaA(n)=0
  \tag{H2}
  \label{eq:beta}
\end{align}
and
\begin{align}
  \sum_{\substack{
      \uu'\in\U^n\setminus\{\uu\}
      \\
      p_{\sfA_n}(\{A: A\uu = A\uu'\})>\frac{\alphaA(n)}{|\im\A_n|}
  }}
  p_{\sfA_n}\lrsb{\lrb{A: A\uu = A\uu'}}
  \leq
  \betaA(n)
  \tag{H3}
  \label{eq:hash}
\end{align}
for any $\uu\in\U^n$.
Throughout this paper, we omit the dependence on $n$ of $\A$, $A$,
$\sfA$, $\alphaA$, and $\betaA$.
\end{df}

Let us remark on the conditions (\ref{eq:alpha})--(\ref{eq:hash}).
These conditions require that the sum of the all collision probabilities
which is far grater than $1/|\im\A|$ vanishes as the block length goes
to infinity.

\begin{rem}
The condition
\begin{equation*}
  \limn \frac1n\log\frac{|\bU_n|}{|\im\A_n|}=0
\end{equation*}
is required for an ensemble in~\cite[Def. 1]{HASH}.
We omit this condition because it is unnecessary for the results
presented in this paper.
\end{rem}

It should be noted that $(\bcA,\bpA)$ has a strong $(\one,\zero)$-hash
property when $\A$ is a universal class of hash functions \cite{CW}
and $\pA$ is the uniform distribution on $\A$.
The random bin coding \cite{C75} and the set of all linear functions
\cite{CSI82} are examples of a universal class of hash functions.
The strong hash property of an ensemble of sparse matrices is discussed
in Section~\ref{sec:linear}.

From the following lemma, we have the fact that the above definition of
hash property is stronger than that introduced in \cite{HASH}.
It is proved in Appendix~\ref{sec:wash}.
\begin{lem}
\label{lem:whash}
If $(\bcA,\bpA)$ has a strong $(\aalphaA,\bbetaA)$-hash property,
then
\begin{align}
  \sum_{\substack{
      \uu\in\T
      \\
      \uu'\in\T'
  }}
  \pA\lrsb{\lrb{A: A\uu = A\uu'}}
  &\leq
  |\T\cap\T'|
  +
  \frac{|\T||\T'|\alphaA}{|\im\A|}
  +
  \min\{|\T|,|\T'|\}\betaA
  \label{eq:whash}
\end{align}
for any $\T,\T'\subset\U^n$, that is,
it has a $(\aalphaA,\bbetaA)$-hash property introduced
in~\cite[Definition 1]{HASH}.
\end{lem}

In the following, we review two lemmas of the hash property.
It should be noted that Lemmas~\ref{lem:collision}
and~\ref{lem:saturation} are related to the collision-resistace property
and the saturation property, respectively.
These relations are explained in \cite[Section III]{HASH}.
Let $\A$ be a set of functions $A:\U^n\to\im\A$,
$\pa$ be the uniform distribution on $\im\A$,
where random variables $\sfA$ and $\sfaa$ be mutually independent,
that is,
\begin{align*}
  \pa(\ba)&\equiv
  \begin{cases}
    \frac 1{|\im\A|},&\text{if}\ \ba\in\im\A
    \\
    0,&\text{if}\ \ba\notin\im\A
  \end{cases}
  \\
  \pAa(A,\ba)&=\pA(A)\pa(\ba)
\end{align*}
for any $A$ and $\ba$.

\begin{lem}[{\cite[Lemma 1]{HASH}}]
\label{lem:collision}
If $(\A,\pA)$ satisfies (\ref{eq:whash}),
then
\[
  \pA\lrsb{\lrb{
      A: \lrB{\G\setminus\{\uu\}}\cap\C_A(A\uu)\neq \emptyset
  }}
  \leq 
  \frac{|\G|\alphaA}{|\im\A|} + \betaA.
\]
for all $\G\subset\U^n$  and $\uu\in\U^n$.
\end{lem}

\begin{lem}[{\cite[Lemma 2]{HASH}}]
\label{lem:saturation}
If $(\A,\pA)$ satisfies (\ref{eq:whash})
then
\begin{align*}
  &
  \pAa\lrsb{\lrb{(A,\ba):
      \T\cap\C_{A}(\ba)=\emptyset
  }}
  \leq
  \alphaA-1+\frac{|\im\A|\lrB{\betaA+1}}{|\T|}
\end{align*}
for all $\T\subset\U^n$.
\end{lem}

In the following, we consider the combination of two ensembles,
where functions have the same domain.
It should be noted that the assumption of a strong hash property makes
it unnecessary to assume the linearity of a function for the hash
property of the concatenated ensemble while the linearity of a function
is assumed in \cite{HASH}\cite{HASH-UNIV}\cite{HASH-CRYPT}.
The proof is given in Appendix \ref{sec:hash-AB}.
\begin{lem}
\label{lem:hash-AB}
Let $(\bcA,\bpA{})$ and $(\bcA',\bpAp{})$ be ensembles satisfying a
strong $(\aalphaA{},\bbetaA{})$-hash property and a strong
$(\aalphaAp,\bbetaAp)$-hash property, respectively.
Let $\A\in\bcA{}$ (resp. $\A'\in\bcA'$) be a set of functions
$A:\U^n\to\im\A$ (resp. $A':\U^n\to\im\A'$).
Let $\hcA\equiv\A\times\A'$  and $\hA\equiv(A,A')\in\hcA$ defined as
\[
  \hA\uu\equiv(A\uu,A'\uu)
  \quad\text{for each $\hA\in\hcA$, $\uu\in\U^n$}.
\]
Let $\phA$  be a joint distribution on $\hcA$ defined as
\[
  p_{\sfhA}(A,A')\equiv \pA(A)\pAp{}(A').
\]
Then the ensemble $(\bhcA,\bphA{}{})$ has a strong
$(\aalphahA{},\bbetahA{})$-hash property where $(\alphahA{},\betahA{})$
is defined as
\begin{align*}
  \alphahA{}&\equiv \alphaA{}\alphaAp{}
  \\
  \betahA{}&\equiv \betaA{}+\betaAp{}.
\end{align*}
\end{lem}

In the following, we consider the combination of two ensembles,
where the domains of functions are different.
The following lemmas are essential for the proof of coding theorems
presented in this paper.
For a set $\A$ of functions $A:\U^n\to\im\A$,
a set $\B$ of functions $B:\V^n\to\im\B$,
let $p_{\sfaa}$ and $p_{\sfbb}$ the uniform distributions on $\im\A$ and
$\im\B$, respectively, where random variables
$\{\sfA,\sfB,\sfaa,\sfbb\}$ are mutually independent, that is,
\[
  p_{\sfA\sfB\sfaa\sfbb}(A,B,\ba,\bb)=
  \pA(A)\pB(B)p_{\sfaa}(\ba)p_{\sfbb}(\bb)
\]
for any $A$, $B$, $\ba$, and $\bb$.
\begin{lem}
\label{lem:hash-AxB-CRP}
If $(\A,\pA)$ and $(\B,\pB)$
satisfy (\ref{eq:hash}), then
\begin{align}
  &p_{\sfA\sfB}\lrsb{\lrb{
      (A,B):
      \lrB{\G\setminus\{(\uu,\vv)\}}\cap\lrB{\C_A(A\uu)\times\C_B(B\vv)}\neq \emptyset
  }}
  \notag
  \\*
  &\leq 
  \frac{|\G|\alphaA\alphaB}{|\im\A||\im\B|}
  +
  \frac{\lrB{\disp\max_{\vv\in\G_{\V}}|\G_{\U|\V}(\vv)|}\alphaA[\betaB+1]}
  {|\im\A|}
  +
  \frac{\lrB{\disp\max_{\uu\in\G_{\U}}|\G_{\V|\U}(\uu)|}\alphaB[\betaA+1]}
  {|\im\B|}
  +\betaA+\betaB+\betaA\betaB
  \label{eq:AxB-CRP}
\end{align}
for all $\G\subset\U^n\times\V^n$ and $(\uu,\vv)\in\U^n\times\V^n$.
\end{lem}
\begin{lem}
\label{lem:hash-AxB-SP}
If $(\A,\pA)$ and $(\B,\pB)$
satisfy (\ref{eq:hash}), then
\begin{align}
  &p_{\sfA\sfB\sfaa\sfbb}\lrsb{\lrb{
      (A,B,\ba,\bb):
      \T\cap\lrB{\C_A(\ba)\times\C_B(\bb)}=\emptyset
  }}
  \notag
  \\*
  \begin{split}
    &
    \leq
    \alphaA\alphaB-1
    +
    \frac{|\im\B|\lrB{\disp\max_{\vv\in\T_{\V}}|\T_{\U|\V}(\vv)|}\alphaA[\betaB+1]}
    {|\T|}
    +
    \frac{|\im\A|\lrB{\disp\max_{\uu\in\T_{\U}}|\T_{\V|\U}(\uu)|}\alphaB[\betaA+1]}
    {|\T|}
    \\*
    &\quad
    +\frac{|\im\A||\im\B|\lrB{\betaA+\betaB+\betaA\betaB+1}}
    {|\T|}
  \end{split}
  \label{eq:AxB-SP}
\end{align}
for all $\T\subset\U^n\times\V^n$.
\end{lem}

It should be noted that Lemmas~\ref{lem:hash-AxB-CRP}
and~\ref{lem:hash-AxB-SP} are related to the collision-resistance
property (Lemma~\ref{lem:collision}) and the saturation property
(Lemma~\ref{lem:saturation}), respectively.
Proof are given in Appendix~\ref{sec:hash-AxB}.

\subsection{Hash Property for Ensembles of Matrices}
\label{sec:linear}
In the following, we discuss the hash property for an ensemble of matrices.

It has been discussed in the last section that the uniform distribution
on the set of all linear functions has a strong $(\one,\zero)$-hash
property because it is a universal class of hash functions.
In the following, we introduce another ensemble of matrices.

First, we introduce the average spectrum of an ensemble of
matrices given in~\cite{BB04}.
Let $\U$ be a finite field and  $\A$ be a set of linear functions
$A:\U^n\to\U^{l}$.
It should be noted that $A$ can be represented by a $l\times n$ matrix.

Let $\bt(\uu)$ be the type\footnote{In \cite{HASH}, it is called
  a histogram which is characterized by the number $n\nu_{\uu}$ of 
  occurrences of each symbol in the sequence $\uu$.
  The type and the histogram is essentially the same.}
of  $\uu\in\U^n$,
which is characterized by the empirical probability distribution
$\nu_{\uu}$ of the sequence $\uu$.
Let $\cH$ be a set of all types of length $n$ except $\bt(\zero)$,
where $\zero$ is the zero vector.
For a probability distribution $\pA$ on a set of $l\times n$ matrices
and a type $\bt$, let $S(\pA,\bt)$ be defined as
\begin{gather*}
  S(\pA,\bt)
  \equiv
  \sum_{A\in\A}\pA(A)|\{\uu\in\U^n: A\uu=\zero, \bt(\uu)=\bt\}|,
\end{gather*}
which is called the expected number of codewords that have type $\bt$
in the context of linear codes.
For given $\hcH_{\sfA}\subset\cH$, we define $\alphaA(n)$ and
$\betaA(n)$ as
\begin{align}
  \alphaA(n)
  &\equiv
  \frac{|\im\A|}{|\U|^{l}}\cdot\max_{\bt\in \hcH_{\sfA}}
  \frac {S(\pA,\bt)}{S(u_{\A},\bt)}
  \label{eq:alpha-linear}
  \\
  \betaA(n)
  &\equiv
  \sum_{\bt\in \cH\setminus\hcH_{\sfA}}S(\pA,\bt),
  \label{eq:beta-linear}
\end{align}
where $u_{\A}$ denotes the uniform distribution on the set of all
$l\times n$ matrices.

We have the following theorem.
The proof is given in Section \ref{sec:proof-linear}.
\begin{thm}
\label{thm:hash-linear}
Let $(\bcA,\bpA)$ be an ensemble of matrices and assume that
$\pA\lrsb{\lrb{A: A\uu=\zero}}$ depends on $\uu$ only through the type
$\bt(\uu)$.
If $(\aalphaA,\bbetaA)$, defined by (\ref{eq:alpha-linear}) and
(\ref{eq:beta-linear}), satisfies (\ref{eq:alpha}) and (\ref{eq:beta}),
then $(\bcA,\bpA)$ has a strong $(\aalphaA,\bbetaA)$-hash property.
\end{thm}

Next, we introduce the ensemble of $q$-ary sparse matrices introduced
in~\cite{HASH}, which is the $q$-ary extension of the ensemble proposed
in~\cite{Mac99}.
Let $\U\equiv\GFq$ and $l\equiv nR$ for given $0<R<1$.
We generate an $l\times n$ matrix $A$ with the following procedure,
where at most $\tau$ random nonzero elements are introduced in every
row.
\begin{enumerate}
  \item Start from an all-zero matrix.
  \item For each $i\in\{1,\ldots,n\}$, repeat the following
  procedure $\tau$ times:
  \begin{enumerate}
    \item Choose $(j,a)\in\{1,\ldots,l\}\times[\GFq\setminus\{0\}]$
    uniformly at random.
    \item Add\footnote{It should be noted that
      $(j,i)$-element of matrix is not overwritten by $a$
      when the same $j$ is chosen again.}
    $a$ to the  $(j,i)$-element of $A$.
  \end{enumerate}
\end{enumerate}
Assume that $\tau=O(\log n)$ is even and let $(\bcA,\bpA)$ be an
ensemble corresponding to the above procedure.
Let $\hcH_{\sfA}\subset\cH$ be a set of types satisfying the requirement
that the weight (the number of occurrences of non-zero elements) is
large enough.
Let $(\aalphaA,\bbetaA)$ be defined by (\ref{eq:alpha-linear})
and (\ref{eq:beta-linear}).
Then $\alphaA$ measures the difference between the ensemble
$(\bcA,\bpA)$ and the ensemble of all $l\times n$ matrices with respect
to the high-weight part of the average spectrum,
and $\betaA$ provides the upper bound of the probability that the code
$\{\uu\in\U^n: A\uu=\zero\}$ has low-weight codewords.
It is proved in \cite[Theorem 2]{HASH} that $(\aalphaA,\bbetaA)$ satisfy
(\ref{eq:alpha}) and (\ref{eq:beta}) by adopting an appropriate
$\hcH_{\sfA}$.
Then, from Theorem~\ref{thm:hash-linear},
we have the fact that this ensemble has a strong
$(\aalphaA,\bbetaA)$-hash property.
It should be noted that the convergence speed of $(\aalphaA,\bbetaA)$
depends on how fast $\tau$ grows in relation to the block length.
The analysis of $(\aalphaA,\bbetaA)$ is given in the proof
of~\cite[Theorem 2]{HASH}.

\section{Slepian-Wolf Source Coding}
In this section, we consider the Slepian-Wolf source coding illustrated
in Fig.~\ref{fig:sw}.
The achievable rate region for this problem is given by
\begin{align*}
  R_X &\geq H(X|Y)
  \\
  R_Y &\geq H(Y|X)
  \\
  R_X+R_Y &\geq H(X,Y),
\end{align*}
where $(R_X,R_Y)$ denotes an encoding rate pair.

The achievability of the Slepian-Wolf source coding is proved
in~\cite{C75} and \cite{CSI82} for an ensemble of bin-coding and all
$q$-ary linear matrices, respectively.
The construction of encoders using sparse matrices is studied
in \cite{M02}\cite{SPR02} and the achievability is proved
in~\cite{SWLDPC} by using ML decoding.
The aim of this section is to demonstrate the proof of the coding
theorem based on the hash property.
The proof is given in Section \ref{sec:proof-sw}.

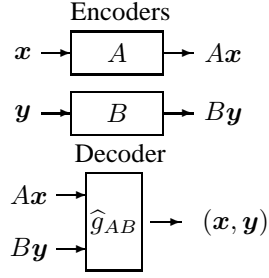
\begin{figure}[t]
  \begin{center}
    \unitlength 0.4mm
    \begin{picture}(176,47)(0,0)
      \put(82,41){\makebox(0,0){Encoders}}
      \put(50,27){\makebox(0,0){$\xx$}}
      \put(56,27){\vector(1,0){10}}
      \put(66,20){\framebox(30,14){$A$}}
      \put(96,27){\vector(1,0){10}}
      \put(116,27){\makebox(0,0){$A\xx$}}
      \put(50,7){\makebox(0,0){$\yy$}}
      \put(56,7){\vector(1,0){10}}
      \put(66,0){\framebox(30,14){$B$}}
      \put(96,7){\vector(1,0){10}}
      \put(116,7){\makebox(0,0){$B\yy$}}
    \end{picture}
    \\
    \begin{picture}(176,45)(0,0)
      \put(82,39){\makebox(0,0){Decoder}}
      \put(52,25){\makebox(0,0){$A\xx$}}
      \put(61,25){\vector(1,0){10}}
      \put(52,7){\makebox(0,0){$B\yy$}}
      \put(61,7){\vector(1,0){10}}
      \put(71,0){\framebox(18,32){$\hg_{AB}$}}
      \put(92,16){\vector(1,0){10}}
      \put(120,16){\makebox(0,0){$(\xx,\yy)$}}
    \end{picture}
  \end{center}
  \caption{Construction of Slepian-Wolf Source Code}
  \label{fig:sw-code}
\end{figure}

We fix functions
\begin{align*}
  A&:\X^n\to\im\A
  \\
  B&:\Y^n\to\im\B,
\end{align*}
which are available for constructing encoders and a decoder.
We define the encoders and the decoder (illustrated in Fig.\ \ref{fig:sw-code})
\begin{align*}
  \Encoder_X&:\X^n\to\im\A
  \\
  \Encoder_Y&:\Y^n\to\im\B
  \\
  \Decoder&:\im\A\times\im\B\to\X^n\times\Y^n
\end{align*}
as
\begin{align*}
  \Encoder_X(\xx)&\equiv A\xx
  \\
  \Encoder_Y(\yy)&\equiv B\yy
  \\
  \Decoder(\ba,\bb)
  &\equiv \hg_{AB}(\ba,\bb),
\end{align*}
where
\begin{align*}
  \hg_{AB}(\ba,\bb)
  &\equiv
  \arg\min_{(\xx',\yy')\in\C_{A}(\ba)\times\C_B(\bb)}
  D(\nu_{\xx'\yy'}\|\mu_{XY}).
\end{align*}
It should be noted that the construction is analogous to the coset
encoding/decoding in the case when $A$ and $B$ are linear functions.

The encoding rate pair $(R_X,R_Y)$ is given by
\begin{align*}
  R_X\equiv\frac{\log|\im\A|}n
  \\
  R_Y\equiv\frac{\log|\im\B|}n
\end{align*}
and the error probability $\Error_{XY}(A,B)$ is given by
\begin{equation*}
  \Error_{XY}(A,B)
  \equiv
  \mu_{XY}\lrsb{\lrb{
      (\xx,\yy): \Decoder(\Encoder_X(\xx),\Encoder_Y(\yy))\neq (\xx,\yy)
  }}.
\end{equation*}
We have the following theorem.
It should be noted that $\X$ and $\Y$ are allowed to be non-binary
and the correlation of the two sources is allowed to be asymmetric.
\begin{thm}
\label{thm:sw}
Assume that $(\bcA,\bpA)$ and $(\bcB,\bpB)$
have a strong hash property.
Let  $(X,Y)$ be a pair of stationary memoryless sources.
If $(R_X,R_Y)$ satisfies
\begin{align}
  R_X &> H(X|Y)
  \label{eq:swx}
  \\
  R_Y &> H(Y|X)
  \label{eq:swy}
  \\
  R_X+R_Y &> H(X,Y),
  \label{eq:swxy}
\end{align}
then for any $\delta>0$
and all sufficiently large $n$
there are functions (sparse matrices) $A\in\A$ and $B\in\B$ such that
\begin{align*}
  &\Error_{XY}(A,B)
  \leq
  \delta.
\end{align*}
\end{thm}

\begin{rem}
Instead of a maximum-likelihood decoder,
we use a minimum-divergence decoder,
which is compatible with typical-set decoding.
We can replace the minimum-divergence decoder by
\begin{align}
  g'_{AB}(\ba,\bb)
  &\equiv
  \arg\max_{\substack{
      (\xx',\yy')\in\C_{A}(\ba)\times\C_B(\bb)
      \\
      \xx'\in\T_{X,\gamma}
      \\
      \yy'\in\T_{Y,\gamma}
  }}
  \mu_{XY}(\xx',\yy'),
  \label{eq:ML}
\end{align}
where $\gamma>0$ should be defined properly.
The reason for introducing the conditions $\xx'\in\T_{X,\gamma}$ and 
$\yy'\in\T_{Y,\gamma}$ is given at the end of Section~\ref{sec:proof-sw}.
\end{rem}

\section{Broadcast Channel Coding}

In this section we consider the broadcast channel coding problem
illustrated in Fig.~\ref{fig:broadcast}.
A broadcast channel is characterized by the conditional probability
distribution $\mu_{YZ|X}$,
where $X$ is a random variable corresponding to the channel input of a
sender, and $(Y,Z)$ is a pair of random variables corresponding to
the channel outputs of respective receivers.
In the following, we consider the case when a sender transmits
two independent messages with nothing in common to two receivers.
It is known that if a rate pair $(R_Y,R_Z)$ satisfies
\begin{align}
  R_Y &\leq I(U;Y)
  \label{eq:RYc}
  \\
  R_Z &\leq I(V;Z)
  \label{eq:RZc}
  \\
  R_Y+R_Z &\leq I(U;Y)+I(V;Z)-I(U;V)
  \label{eq:RYRZc}
\end{align}
for some joint probability distribution $\mu_{UVX}$ on $\U\times\V\times\X$,
then there is a code for this channel such that the decoding error
probability goes to zero as the block length goes to infinity,
where the joint distribution of random variable $(U,V,X,Y,Z)$ is given
by
\begin{align}
  \mu_{UVXYZ}(u,v,x,y,z)
  \equiv \mu_{YZ|X}(y,z|x)\mu_{UVX}(u,v,x).
  \label{eq:markov-broadcast}
\end{align}
This type of achievable region is derived in \cite{M79} and a simpler
proof is given in \cite{EM81}.
Furthermore, it is shown in \cite[Theorem 1]{GA09} that it is sufficient
to consider a joint probability distribution $\mu_{UVX}$ satisfying
$|\U|\leq|\X|$, $|\V|\leq|\X|$, and $H(X|UV)=0$.
In the following, we assume that $|\U|\leq|\X|$, $|\V|\leq|\X|$,
and there is a function $f:\U\times\V\to\X$ such that 
\begin{equation}
  \mu_{UVX}(u,v,x)
  =
  \mu_{UV}(u,v)\chi(x=f(u,v)),
  \label{eq:UVfUV}
\end{equation}
where (\ref{eq:UVfUV}) is equivalent to $H(X|UV)=0$.

\begin{figure}[t]
  \begin{center}
    \unitlength 0.40mm
    \begin{picture}(186,99)(0,6)
      \put(94,89){\makebox(0,0){Encoder}}
      \put(20,6){\framebox(148,75){}}
      \put(30,66){\makebox(0,0){$\ba$}}
      \put(35,66){\vector(1,0){10}}
      \put(0,51){\makebox(0,0){$\mm$}}
      \put(10,51){\vector(1,0){35}}
      \put(30,36){\makebox(0,0){$\bb$}}
      \put(35,36){\vector(1,0){10}}
      \put(0,21){\makebox(0,0){$\ww$}}
      \put(10,21){\vector(1,0){35}}
      \put(45,13){\framebox(32,61){$\hg_{AA'BB'}$}}
      \put(77,43.5){\vector(1,0){10}}
      \put(102,43.5){\makebox(0,0){$(\uu,\vv)$}}
      \put(117,43.5){\vector(1,0){10}}
      \put(127,36.5){\framebox(30,14){$\ff$}}
      \put(157,43.5){\vector(1,0){20}}
      \put(183,43.5){\makebox(0,0){$\xx$}}
    \end{picture}
    \\
    \begin{picture}(156,70)(0,0)
      \put(82,60){\makebox(0,0){Decoders}}
      \put(30,35){\makebox(0,0){$\ba$}}
      \put(35,35){\vector(1,0){10}}
      \put(0,17){\makebox(0,0){$\yy$}}
      \put(10,17){\vector(1,0){35}}
      \put(45,10){\framebox(18,32){$g_A$}}
      \put(63,26){\vector(1,0){10}}
      \put(83,26){\makebox(0,0){$\uu$}}
      \put(93,26){\vector(1,0){10}}
      \put(103,19){\framebox(30,14){$A'$}}
      \put(133,26){\vector(1,0){20}}
      \put(167,26){\makebox(0,0){$\mm$}}
      \put(20,0){\framebox(124,52){}}
    \end{picture}\\
    \begin{picture}(156,60)(0,0)
      \put(30,35){\makebox(0,0){$\bb$}}
      \put(35,35){\vector(1,0){10}}
      \put(0,17){\makebox(0,0){$\zz$}}
      \put(10,17){\vector(1,0){35}}
      \put(45,10){\framebox(18,32){$g_B$}}
      \put(63,26){\vector(1,0){10}}
      \put(83,26){\makebox(0,0){$\vv$}}
      \put(93,26){\vector(1,0){10}}
      \put(103,19){\framebox(30,14){$B'$}}
      \put(133,26){\vector(1,0){20}}
      \put(167,26){\makebox(0,0){$\ww$}}
      \put(20,0){\framebox(124,52){}}
    \end{picture}
  \end{center}
  \caption{Construction of Broadcast Channel Code}
  \label{fig:broadcast-code}
\end{figure}
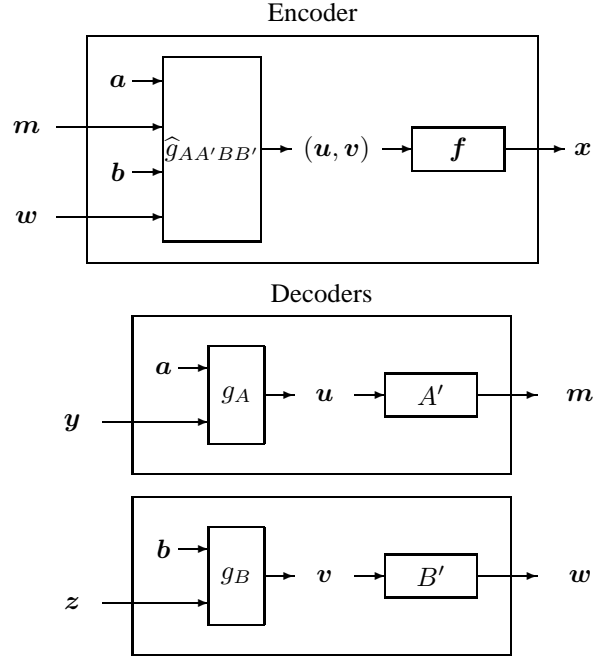

Let $(R_Y,R_Z)$  be a pair of encoding rates.
In the following, we assume that $(R_Y,R_Z)$, $(r_Y,r_Z)$, and $\e$
satisfy
\begin{align}
  r_Y&>H(U|Y)
  \label{eq:rY}
  \\
  r_Z&>H(V|Z)
  \label{eq:rZ}
  \\
  r_Y+R_Y&<H(U)-\e
  \label{eq:RYe}
  \\
  r_Z+R_Z&<H(V)-\e
  \label{eq:RZe}
  \\
  r_Y+R_Y+r_Z+R_Z &< H(U,V)
  \\
  r_Y+R_Y+r_Z+R_Z &> H(U,V)-\e
  \label{eq:rYRYrZRZ-lower}
\end{align}
for given $\mu_{YZ|X}$, $\mu_{UV}$, and $f$.
It should be noted that there are $(r_Y,r_Z)$ and $\e$ when $(R_Y,R_Z)$
satisfies
\begin{align}
  R_Y&< I(U;Y)
  \label{eq:RY}
  \\
  R_Z&< I(V;Z)
  \label{eq:RZ}
  \\
  R_Y+R_Z&< I(U;Y)+I(V;Z)-I(U;V).
  \label{eq:RYRZ}
\end{align}
This fact is shown from the relations
\begin{align*}
  H(U)-H(U|Y)&=I(U;Y)
  \\
  H(V)-H(V|Z)&=I(V;Z)
  \\
  H(U,V)-H(U|Y)-H(V|Z)
  &=I(U;Y)+I(V;Z)
  -I(U;V).
\end{align*}
We fix functions
\begin{align*}
  A&:\U^n\to\im\A
  \\
  A'&:\U^n\to\im\A'
  \\
  B&:\V^n\to\im\B
  \\
  B'&:\V^n\to\im\B'
\end{align*}
and vectors
\begin{align*}
  \ba&\in\im\A
  \\
  \bb&\in\im\B
\end{align*}
available for an encoder and decoders satisfying
\begin{align*}
  r_Y&=\frac{\log|\im\A|}n
  \\
  R_Y&=\frac{\log|\im\A'|}n
  \\
  r_Z&=\frac{\log|\im\B|}n
  \\
  R_Z&=\frac{\log|\im\B'|}n.
\end{align*}
We define the encoder and the decoders
\begin{align*}
  \Encoder&:\im\A'\times\im\B'\to\X^n
  \\
  \Decoder_Y&:\Y^n\to\im\A'
  \\
  \Decoder_Z&:\Z^n\to\im\B'
\end{align*}
as
\begin{align*}
  \Encoder(\mm,\ww)
  &\equiv \ff(g_{AA'BB'}(\ba,\mm,\bb,\ww))
  \\
  \Decoder_Y(\yy)
  &\equiv A'g_A(\ba|\yy)
  \\
  \Decoder_Z(\zz)
  &\equiv B'g_B(\bb|\zz),
\end{align*}
where
\begin{align*}
  \ff(\uu,\vv)&\equiv (f(u_1,v_1),\ldots,f(u_n,v_n))
  \\
  \hg_{AA'BB'}(\ba,\mm,\bb,\ww)
  &\equiv
  \arg\min_{\substack{(\uu',\vv'):\\
      \uu'\in\C_{A}(\ba)\cap\C_{A'}(\mm)\\
      \vv'\in\C_{B}(\bb)\cap\C_{B'}(\ww)
  }}
  D(\nu_{\uu'\vv'}\|\mu_{UV})
  \\
  g_{A}(\ba|\yy)
  &\equiv\arg\max_{\uu'\in\C_A(\ba)}\mu_{U|Y}(\uu'|\yy)
  \\
  g_{B}(\bb|\zz)
  &\equiv\arg\max_{\vv'\in\C_B(\bb)}\mu_{V|Z}(\vv'|\zz).
\end{align*}

Let $\sfmm$ and $\sfww$ be random variables corresponding to messages
$\mm$ and $\ww$, respectively,
where the probability distributions $p_{\sfmm}$ and $p_{\sfww}$ are given by
\begin{align}
  p_{\sfmm}(\mm)
  &\equiv
  \begin{cases}
    \frac 1{|\im\A'|}
    &\text{if}\ \mm\in\im\A'
    \\
    0,
    &\text{if}\ \mm\notin\im\A'
  \end{cases}
  \label{eq:pM}
  \\
  p_{\sfww}(\ww)
  &\equiv
  \begin{cases}
    \frac 1{|\im\B'|}
    &\text{if}\ \ww\in\im\B'
    \\
    0,
    &\text{if}\ \ww\notin\im\B'
  \end{cases}
  \label{eq:pW}
\end{align}
and the joint distribution $p_{\sfmm\sfww YZ}$ of the messages and
outputs is given by
\begin{align*}
  p_{\sfmm\sfww YZ}(\mm,\ww,\yy,\zz)
  &\equiv
  \mu_{YZ|X}(\yy,\zz|\Encoder(\mm,\ww))p_M(\mm)p_W(\ww).
\end{align*}

Let us remark an intuitive interpretation of the code construction,
which is illustrated in Fig.~\ref{fig:broadcast-code}.
Assume that $\ba$ and $\bb$ are shared by the encoder and the decoder.
For $\ba$, $\bb$, and messages $\mm$, $\ww$,
the function $\hg_{AA'BB'}$ generates a pair of typical sequence
$(\uu,\vv)\in\T_{UV,\gamma}$ and it is converted to a channel input
$\xx$ by using a function $f$, where $(\uu,\vv,\xx)$ is jointly-typical.
One decoder reproduces $\uu$ by using $g_A$ from $\ba$ and a channel
output $\yy$ and another decoder reproduces $\vv$ by using $g_B$ from
$\bb$ and a channel output $\zz$.
Since $A'\uu=\mm$, $B'\vv=\ww$, and $(\uu,\vv,\xx,\yy)$ is jointly
typical, the decoding succeeds if the amount of information of $\ba$
(resp. $\bb$) is greater than $H(U|Y)$ (resp. $H(V|Z)$) to satisfy the
collision-resistance property.
On the other hand, the rate of $\ba$, $\mm$, $\bb$, and $\ww$ should
satisfy the saturation property to find a jointly typical sequences
$(\uu,\vv)$ with probability close to $1$, that is,
the right hand side of (\ref{eq:AxB-SP}) should tend to zero as the
block length goes to infinity.
Then we can set the encoding rate of $\mm$ and $\ww$ 
satisfying (\ref{eq:RY})--(\ref{eq:RYRZ}).

Let $\Error_{YZ|X}(A,A',B,B',\ba,\bb)$ be the decoding error probability
given by
\begin{align}
  &\Error_{YZ|X}(A,A',B,B',\ba,\bb)
  \notag
  \\*
  &\equiv
  1- \sum_{\mm,\ww,\yy,\zz}\mu_{YZ|X}(\yy,\zz|\Encoder(\mm,\ww))p_M(\mm)p_W(\ww)
  \chi(\Decoder_Y(\yy)=\mm)
  \chi(\Decoder_Z(\zz)=\ww).
  \label{eq:def-error-broadcast}
\end{align}
Then we have the following theorem.
It should be noted that the channel is allowed to be asymmetric and
non-degraded.
\begin{thm}
\label{thm:broadcast}
Let $\mu_{YZ|X}$ be the conditional probability distribution of a
stationary memoryless channel and $\mu_{UVXYZ}$ be defined by
(\ref{eq:markov-broadcast}) and (\ref{eq:UVfUV}) for a given joint
probability distribution $\mu_{UV}$ and a function $f$.
For a given $(R_Y,R_Z)$, $(r_Y,r_Z)$, and $\e$ satisfying
(\ref{eq:rY})--(\ref{eq:rYRYrZRZ-lower}), assume that ensembles
$(\bcA,\bpA)$, $(\bcA',\bpAp)$, $(\bcB,\bpB)$, and $(\bcB',\bpBp)$
have a strong hash property.
Then, for any $\delta>0$ and all sufficiently large $n$,
there are functions (sparse matrices) $A\in\A$, $A'\in\A'$, $B\in\B$,
$B'\in\B'$ and vectors $\ba\in\im\A$, $\bb\in\im\B$ such that
\begin{gather}
  \Rate_Y > R_Y-\delta
  \label{eq:rateY}
  \\
  \Rate_Z > R_Z-\delta
  \label{eq:rateZ}
  \\
  \Error_{YZ|X}(A,A',B,B',\ba,\bb)<\delta
  \label{eq:error}
\end{gather}
By assuming that $\delta\to 0$, the rate of the proposed code is close
to the boundary of the region specified by
(\ref{eq:RYc})--(\ref{eq:RYRZc}) for a given $\mu_{UV}$ and $f$.
\end{thm}

\begin{rem}
It should be noted that the maximum-likelihood decoders $g_A$ and $g_B$
can be replaced by the minimum-divergence decoders $\hg_A$ and $\hg_B$
defined as
\begin{align*}
  \hg_A(\ba|\yy)
  &\equiv\arg\min_{\uu'\in\C_A(\ba)}D(\nu_{\uu'|\yy}\|\mu_{U|Y}|\nu_{\yy})
  \\
  \hg_B(\bb|\zz)
  &\equiv\arg\min_{\vv'\in\C_B(\bb)}D(\nu_{\vv'|\zz}\|\mu_{V|Z}|\nu_{\zz}),
\end{align*}
respectively.
\end{rem}

\begin{rem}
In \cite{EM81}\cite{M79},
the set of strong typical sequences are used for the construction of
a broadcast channel code, where there is no structure of the set of
codewords.
It should be noted that the encoder and the decoder have to share
the large (exponentially in the block length) table which indicates the
correspondence between an index and codewords.
The time complexity of the decoding grows exponentially in the block length.
They represent obstacles to the implementation.
In our construction, almost all codewords have strong typicality but the
set of codewords have a structure specified by functions 
$A$, $A'$, $B$, $B'$ and vectors $\ba$, $\bb$.
When the functions are linear, it is expected that the space and time
complexity can be reduced compared with conventional constructions.
\end{rem}
\begin{rem}
It should be noted that the inner bound presented in the beginning of
this section is not convex in general and we can apply the time sharing
principle to obtain the convex hull of this region.
It should also be noted that, according to \cite[Page 9-36]{EK10},
there is a much larger inner region given by the set of a rate pair
$(R_Y,R_Z)$ satisfying
\begin{align*}
  R_Y &\leq I(S,U;Y)
  \\
  R_Z &\leq I(S,V;Z)
  \\
  R_Y+R_Z &\leq I(S,U;Y)+I(V;Z|S)-I(U;V|S)
  \\
  R_Y+R_Z &\leq I(U;Y|S)+I(S,V;Z)-I(U;V|S),
\end{align*}
for some probability distribution $\mu_{SUV}$ on $\cS\times\U\times\V$
and a function $f:\cS\times\U\times\V\to\X$.
We presented the code construction for the smaller inner region.
because it is a typical application of the extended saturation property
(Lemma~\ref{lem:hash-AxB-SP}).
It is a future challenge to obtain a code construction for the larger
inner region based on the strong hash property.
\end{rem}

\section{Construction of Minimum-Divergence Operation}
\label{sec:binary}

In this section, we introduce the construction of the minimum-divergence
operation
\begin{align*}
  \hg_{\hA\hB}(\haa,\hbb)
  &\equiv
  \arg\min_{(\uu',\vv')\in\C_{\hA}(\haa)\times\C_{\hB}(\hbb)}
  D(\nu_{\uu'\vv'}\|\mu_{UV}),
\end{align*}
which is used to construct the encoder of the broadcast channel
and the decoder of the Slepian-Wolf codes.

The construction of this operation is introduced in \cite{HASH-UNIV}
for a binary alphabet.
However, we have to construct this operation
for an alphabet that is a product space $\U\times\V$.

In the following, we assume that $\{0,1\}$ is a finite field,
$\U=\V=\{0,1\}$ and that $\hA$ and $\hB$ are matrices on $\{0,1\}$.
It should be noted that our construction of broadcast channel codes
can approach the boundary of the achievable region by assuming $|\X|=2$
(see \cite[Theorem 1]{GA09}).

Let $t_{\uu\vv}(u,v)\in\{0,\ldots,n\}$ be the number of occurrences of
$(u,v)\in\{0,1\}^2$ in $(\uu,\vv)\in\{0,1\}^n\times\{0,1\}^n$.
For example, when $n=8$ and
\begin{align*}
  \uu &= 01001010
  \\
  \vv &= 00101001,
\end{align*}
we have
\begin{align*}
  t_{\uu\vv}(0,0)&=3
  \\
  t_{\uu\vv}(0,1)&=2
  \\
  t_{\uu\vv}(1,0)&=2
  \\
  t_{\uu\vv}(1,1)&=1.
\end{align*} 
Let $\bt_{\uu\vv}\equiv\lrb{t_{\uu\vv}(0,0),t_{\uu\vv}(0,1),t_{\uu\vv}(1,0),t_{\uu\vv}(1,1)}$.
We have
\begin{align*}
  \hg_{\hA\hB}(\haa,\hbb)
  &=
  \arg\min_{\substack{
      (\uu_{\bt},\vv_{\bt}):
      \\
      \bt\in\cH
  }}
  D(\nu_{\uu_{\bt}\vv_{\bt}}\|\mu_{UV}),
\end{align*}
where
\begin{align}
  (\uu_{\bt},\vv_{\bt})
  &\equiv
  \begin{cases}
    (\uu,\vv)
    &\text{if}\ \exists(\uu,\vv)\ \text{s.t.}\ \bt_{\uu\vv}=\bt,
    \hA\uu=\haa, \hB\vv=\hbb
    \\
    \text{`error'}
    &\text{otherwise}
  \end{cases} 
  \label{eq:uvt}
  \\
  \cH
  &\equiv\lrb{\bt: t(u,v)\in\{0,\ldots,n\}, \sum_{u,v}t(u,v)=n}
  \notag
\end{align}
and $D(\nu_{\uu_{\bt}\vv_{\bt}}\|\mu_{UV})\equiv\infty$
when the operation (\ref{eq:uvt}) returns an error.
It should be noted that the cardinality of the set $\cH$
is at most $[n+1]^4$.

We can apply the linear programming technique introduced in \cite{FWK05}
to the construction of (\ref{eq:uvt}).
By using the technique introduced in \cite{FWK05},
the conditions $\hA\uu=\haa$ and $\hB\vv=\hbb$
can be replaced by linear inequalities.
In the following, we focus on the condition $\bt_{\uu\vv}=\bt$
for a given $\bt\in\cH$.

Let $\uu\equiv(u_1,\ldots,u_n)$, $\vv\equiv(v_1,\ldots,v_n)$
and $\bt\equiv(t(0,0),t(0,1),t(1,0),t(1,1))$.
First, we introduce the auxiliary variable
$\bs(u,v)\equiv(s_1(u,v),\ldots,s_n(u,v))$ for $(u,v)\in\{0,1\}^2$ defined as
\begin{equation}
  s_i(u,v)\equiv
  \begin{cases}
    1,&\text{if}\ u_i=u, v_i=v
    \\
    0,&\text{otherwise}.
  \end{cases}
  \label{eq:si}
\end{equation}
Then the condition $\bt_{\uu\vv}=\bt$ is equivalent to 
linear equalities
\begin{equation*}
  \sum_{i=1}^n s_i(u,v)=t(u,v)
\end{equation*}
for all $(u,v)\in\{0,1\}^2$.

Next, we show that the relation (\ref{eq:si}) can be replaced by
linear inequalities.
Let $\cS(u,v)$ be defined as
\[
  \cS(u,v)\equiv\{(1,u,v)\}\cup\lrB{\bigcup_{(u',v')\neq(u,v)}\{(0,u',v')\}}.
\]
Then we obtain the convex hull of $\cS(u,v)$
defined by
the set of $(s',u',v')$ that satisfies
the following linear inequalities.
\begin{align*}
  s'&\geq 0
  \\
  u'&\geq 0\quad\text{if}\ u=0
  \\
  u'&\leq 1\quad\text{if}\ u=1
  \\
  v'&\geq 0\quad\text{if}\ v=0
  \\
  v'&\leq 1\quad\text{if}\ v=1
  \\
  s'+[-1]^uu'&\leq 1-u
  \\
  s'+[-1]^vv'&\leq 1-v
  \\
  s'+[-1]^uu'+[-1]^vv'&\geq 1-u-v,
\end{align*}
where constants $u$, $v$ and variables $s'$, $u'$, and $v$
are considered to be real numbers.
For example, when $(u,v)=(0,0)$, we have
\begin{gather*}
  \cS(0,0)=\{(1,0,0),(0,0,1),(0,1,0),(0,1,1)\}
\end{gather*}
and linear inequalities
\begin{align*}
  s'&\geq 0
  \\
  u'&\geq 0
  \\
  v'&\geq 0
  \\
  s'+u'&\leq 1
  \\
  s'+v'&\leq 1
  \\
  s'+u'+v'&\geq 1.
\end{align*}

Finally,
we have $24n+4$ inequalities,
which replace the condition $\bt_{\uu\vv}=\bt$,
defined as
\begin{align*}
  \sum_{i=1}^n s_i(u,v)&=t(u,v)
  \\
  s_i(u,v)&\geq 0
  \\
  u_i&\geq 0\quad\text{if}\ u=0
  \\
  u_i&\leq 1\quad\text{if}\ u=1
  \\
  v_i&\geq 0\quad\text{if}\ v=0
  \\
  v_i&\leq 1\quad\text{if}\ v=1
  \\
  s_i(u,v)+[-1]^uu_i&\leq 1-u
  \\
  s_i(u,v)+[-1]^vv_i&\leq 1-v
  \\
  s_i(u,v)+[-1]^uu_i+[-1]^vv_i&\geq 1-u-v,
\end{align*}
where we take all $i\in\{1,2,\ldots,n\}$ and $(u,v)\in\{0,1\}^2$.

\begin{rem}
It should be noted here that
linear programming frequently finds non-integral solutions for
the operation (\ref{eq:uvt}) even if
$\hA$ and $\hB$ are sparse matrices. 
Let us consider the polytope with vertexes
$\C_{\hA}(\haa)\times\C_{\hB}(\hbb)$.
Then the section of the polytope
cut by the hyperplane defined by $\bt_{\uu\vv}=\bt$
may have many fractional points.
This is one reason why the linear programming finds non-integral
solutions.
In fact,
the original operation introduced
in \cite{HASH-UNIV} has the same problem.
It remains a future challenge to find a good implementation
of the operation (\ref{eq:uvt}).
\end{rem}

\section{Extension to Three or More Terminals}

In this section, we extend our results to three or more terminals.
We use the following notations:
\begin{align*}
  \K&\equiv\{1,2,\ldots,k\}
  \\
  \U_{\K}&\equiv\Prod_{j\in\K}\U_j
  \\
  \uu_{\K}&\equiv\lrb{\lrsb{u_{j,1},u_{j,2},\ldots,u_{j,n}}}_{j\in\K}
  \\
  A_{\K}&\equiv\lrb{A_j}_{j\in\K}
  \intertext{and}
  U_{\J}&\equiv\lrb{U_j}_{j\in\J}
\end{align*}
for each $\J\subset\K$.

\subsection{Multiple Slepian-Wolf Source Code}
In this section, we consider the Slepian-Wolf coding for $k$ correlated
sources.

Let $X_{\K}\in\X_{\K}$ be correlated sources and $\mu_{X_{\K}}$ be the
joint probability distribution of $X_{\K}$.
For each $j\in\K$, we consider an ensemble $(\bcA_j,\bp_{\sfA_j})$
of functions $A_j:\X_j^n\to\im\A_j$.
We assume that function $A_j$ is shared by the $j$-th encoder and the
decoder.
We define the $j$-th encoder as
\[
  \Encoder_j\lrsb{\xx_j}\equiv A_j\xx_j
\]
for each output $\xx_j\in\X_j^n$ of the $j$-th source.
Let $\hg_{A_{\K}}$ be defined as
\begin{align*}
  \hg_{A_{\K}}\lrsb{\ba_{\K}}
  \equiv
  \arg\min_{\xx_{\K}\in\C_{A_{\K}}(\ba_{\K})}
  D\lrsb{\nu_{\xx_{\K}}\|\mu_{X_{\K}}},
\end{align*}
where $\ba_j$ is the codeword of the $j$-th encoder and
\[
  \C_{A_{\K}}(\ba_{\K})\equiv
  \{\uu_{\K}: A_j\uu_j=\ba_j\ \text{for all}\ j\in\K\}.
\]
Then we define the decoder as
\begin{align*}
  \Decoder\lrsb{\ba_{\K}}
  \equiv\hg_{A_{\K}}\lrsb{\ba_{\K}}.
\end{align*}
The construction of the above minimum-divergence decoder
is described in Section~\ref{sec:md-multi}.

For each $j\in\K$,
the encoding rate $R_j$ for the $j$-th encoder is given by
\[
  R_j\equiv\frac{\log|\im\A_j|}n
\]
and the error probability is given by
\begin{align*}
  \Error_X\lrsb{A_{\K}}
  \equiv
  \mu_{X_{\K}}\lrsb{\lrb{
      \xx_{\K}:
      \Decoder\lrsb{
	\lrb{\Encoder_j\lrsb{\xx_j}}_{j\in\K}
      }
      \neq\xx^{\K}
  }}.
\end{align*}

We have the following theorem.
\begin{thm}
\label{thm:multi-sw}
Assume that $\lrb{R_j}_{j\in\K}$ satisfies
\begin{align*}
  \sum_{j\in\J}R_j>H(X_{\J}|X_{\J^c})
\end{align*}
for all $\J\subset\K$ except the empty set, and an ensemble
$(\bcA_j,\bp_{\sfA_j})$ has an $(\aalpha_{\sfA_j},\bbeta_{\sfA_j})$-hash
property for all $j\in\K$.
Then for any $\delta>0$ and all sufficiently large $n$ there are
functions (sparse matrices) $A_{\K}\equiv\lrb{A_j}_{j\in\K}$ such that
\[
  \Error_X\lrsb{A_{\K}}\leq\delta.
\]
\end{thm}

\subsection{Multiple Output Broadcast Channel Code}
In this section, we consider the $k$-output broadcast channel coding.

Let $\mu_{Y_{\K}|X}$ be the conditional distribution of a broadcast
channel with input $X\in\X$ and $k$ outputs $Y_{\K}\in\Y_{\K}$.

Let $U_{\K}\in\U_{\K}$ be a multiple random variable and
$\mu_{U_{\K}}$ be the joint probability distribution of $U_{\K}$.
Let $F:\U_{\K}\to\X$ be a function that is allowed to be
non-deterministic.
The joint distribution $\mu_{U_{\K}XY_{\K}}$ of
$\lrsb{U_{\K},X,Y_{\K}}$ is given by
\begin{align*}
  \mu_{U_{\K}XY_{\K}}\lrsb{u_{\K},x,y_{\K}}
  \equiv
  \mu_{Y_{\K}|X}(y_{\K}|x)\mu_{X|U_{\K}}(x|u_{\K})\mu_{U_{\K}}(u_{\K}),
\end{align*}
where $\mu_{X|U_{\K}}$ represents the transition probability distribution
of the function $F$.

For each $j\in\K$, we consider an ensemble $(\bcA_j,\bp_{\sfA_j})$
of functions $A_j:\U_j^n\to\im\A_j$, where $\U_j$ is an alphabet of the
source $U_j$.
We also consider an ensemble $(\bcA'_j,\bp_{\sfA'_j})$ of functions
$A'_j:\U_j^n\to\im\A'_j$.
Let $\sfmm_j$ be the random variable corresponding to the $j$-th message,
where the probability distribution $p_{\sfmm_j}$ are given by
\[
  p_{\sfmm_j}(\mm_j)
  \equiv
  \begin{cases}
    \frac 1{|\im\A'_j|}
    &\text{if}\ \mm_j\in\im\A'_j
    \\
    0,
    &\text{if}\ \mm_j\notin\im\A'_j
  \end{cases}
\]
for each $j\in\K$.
We assume that functions
$A_j$, $A'_j$ and a vector
$\ba_j\in\im\A_j$
are shared by the encoder and the $j$-th decoder.
For a multiple message
$\mm_{\K}\in\Prod_{j\in\K}\im\A'_j$,
we define the encoder 
$\Encoder:\Prod_{j\in\K}\im\A'_j\to\X^n$ as
\begin{gather*}
  \hg_{A_{\K}A'_{\K}}(\mm_{\K},\ba_{\K})
  \equiv
  \arg\min_{\uu_{\K}\in\C_{A_{\K}A'_{\K}}(\mm_{\K},\ba_{\K})}
  D(\nu_{\uu_{\K}}\|\mu_{U_{\K}})
  \\
  \FF\lrsb{\uu_{\K}}
  \equiv\lrsb{F(u_{\K,1}),\ldots,F(u_{\K,n})}
  \\
  \Encoder\lrsb{\mm_{\K}}
  \equiv \FF\lrsb{
    \hg_{A_{\K}A'_{\K}}(\mm_{\K},\ba_{\K})
  },
\end{gather*}
where $u_{\K,i}\equiv\lrb{u_{j,i}}_{j\in\K}$ for each
$i\in\{1,\ldots,n\}$
and
\[
  \C_{A_{\K}A'_{\K}}(\mm_{\K},\ba_{\K})\equiv
  \lrb{\uu_{\K}:
    A_j\uu_j=\ba_j, A'_j\uu_j=\mm_j
    \quad \text{for all}\ j\in\K
  }.
\]
The construction of the above minimum-divergence decoder
is described in Section~\ref{sec:md-multi}.
We define the $j$-th decoder
$\Decoder_j:\Y_j^n\to\im\A'_j$ as
\begin{gather*}
  g_{A_j}(\ba_j|\yy_j)
  \equiv
  \arg\max_{\uu_j\in\C_{A_j}(\ba_j)}
  \mu_{U_j|Y_j}(\uu_j|\yy_j)
  \\
  \Decoder_j(\yy_j)\equiv
  A'_jg_{A_j}(\ba_j|\yy_j),
\end{gather*}
for each output $\yy_j\in\Y_j^n$.
The error probability is given by
\begin{align*}
  \Error_{Y_{\K}|X}\lrsb{A_{\K},A'_{\K},\ba_{\K}}
  \equiv 1
  -\sum_{\yy_{\K}}
  \mu_{Y_{\K}|X}(\yy_{\K}|\Encoder(\mm_{\K}))
  \prod_{j\in\K}\lrB{
    p_{\sfmm_j}(\mm_j)
    \chi\lrsb{\Decoder_j(\yy_j)=\mm_j)}
}.
\end{align*}

We have the following theorem.
\begin{thm}
\label{thm:multi-bc}
For each $j\in\K$, let $R_j$ be the encoding rate for the $j$-th encoder.
Assume that $\lrb{(r_j,R_j,\e_j)}_{j\in\K}$ and $\e$ satisfy
\begin{gather}
r_j
=
\frac{\log|\im\A_j|}n
\label{eq:mac-multi-rj}
\\
R_j
=
\frac{\log|\im\A'_j|}n
\\
r_j>H(U_j|Y_j)
\\
\sum_{j\in\J}\lrB{
  R_j+r_j
}
<H(U_{\J})-\e
\\
H(U_{\K})-\e
<
\sum_{j\in\K}\lrB{
  R_j+r_j
}
<
H(U_{\K})
\label{eq:mac-multi-HUK}
\end{gather}
for all non-empty set $\J\subsetneq\K$ and $j\in\K$.
Furthermore, assume that an ensemble $(\bcA_j,\bp_{\sfA_j})$
(resp.\ $(\bcA'_j,\bp_{\sfA'_j})$) has an
$(\aalpha_{\sfA_j},\bbeta_{\sfA_j})$-hash 
(resp.\  $(\aalpha_{\sfA'_j},\bbeta_{\sfA'_j})$-hash)
property for all $j\in\K$.
Then for any $\delta>0$ and all sufficiently large $n$
there are functions (sparse matrices) and vectors
$\lrsb{A_{\K},A'_{\K},\ba_{\K}}\equiv\lrb{\lrsb{A_j,A'_j,\ba_j}}_{j\in\K}$
such that
\[
\Error_{Y_{\K}|X}\lrsb{A_{\K},A'_{\K},\ba_{\K}}
\leq\delta.
\]
\end{thm}

It should be noted that there are $\{r_j\}_{j\in\K}$ and
$\e$ satisfying (\ref{eq:mac-multi-rj})--(\ref{eq:mac-multi-HUK})
if $\{R_j\}_{j\in\K}$ satisfy
\begin{align}
\sum_{j\in\J}R_j
&<H(U_{\J})-\sum_{j\in\J}H(U_j|Y_j)
\notag
\\
&\ =\sum_{j\in\J}I(U_j;Y_j)-\lrB{\sum_{j\in\J}H(U_j)-H(U_{\J})}
\label{eq:mac-multi-R}
\end{align}
for all non-empty set $\J\subset\K$.
The condition (\ref{eq:mac-multi-R})
is equivalent to the conditions
\begin{align*}
R_j&<I(U_j;Y_j)\ \text{for}\ j\in\{1,2,3\}
\\
R_1+R_2&<I(U_1;Y_1)+I(U_2;Y_2)-I(U_1;U_2)
\\
R_1+R_3&<I(U_1;Y_1)+I(U_3;Y_3)-I(U_1;U_3)
\\
R_2+R_3&<I(U_2;Y_2)+I(U_3;Y_3)-I(U_2;U_3)
\\
R_1+R_2+R_3
&<I(U_1;Y_1)+I(U_2;Y_2)+I(U_3;Y_3)
-I(U_1;U_2)-I(U_1,U_2;U_3)
\end{align*}
specified in \cite[Page 9-47]{EK10}, which considers the case of $k=3$.

\subsection{Fundamental Lemmas for
Theorems \ref{thm:multi-sw} and \ref{thm:multi-bc}}

We can prove Theorems \ref{thm:multi-sw} and \ref{thm:multi-bc}
in a similar way to Theorems \ref{thm:sw} and \ref{thm:broadcast}
by using the following lemmas proved in Appendix
\ref{sec:proof-multi}.
To shorten the description of the flollwing lemma,
we use the following abbreviation
\begin{align}
\lrbar{\T_{\J|\J^c}}
&\equiv
\begin{cases}
  1, &\text{if}\ \J=\emptyset
  \\
  |\T|, &\text{if}\ \J=\K
  \\
  \disp\max_{\uu_{\J^c}\in\T_{\U_{\J^c}}}
  \lrbar{\T_{\U_{\J}|\U_{\J^c}}\lrsb{\uu_{\J^c}}}
  &\text{otherwise}.
\end{cases}
\label{eq:maxT}
\end{align}
for $\T\subset[\U_{\K}]^n$ and $\J\subset\K$.
It should be noted that the expression $\lrbar{\T_{\J|\J^c}}$ 
does not represent the cardinality of the set $\T_{\J|\J^c}$.

\begin{lem}
\label{lem:multi-lemma}
For each $j\in\K$, let $\A_j$ be a set
of functions $A_j:\U_j^n\to\im\A_j$
and $p_{\sfA_j}$ be the probability distribution on $\A_j$,
where $(\A_j,p_{\sfA_j})$ satisfies (\ref{eq:hash}).
Let the joint distribution $p_{\sfA_{\K}\sfaa_{\K}}$ be defined as
\[
p_{\sfA_{\K}\sfaa_{\K}}
(A_{\K},\ba_{\K})
\equiv
\prod_{j\in\K}p_{\sfA_j}(A_j)p_{\sfaa_j}(\ba_j)
\]
for each $\{A_j,\ba_j\}_{j\in\K}$.
For each $\J\subset\K$,
let $\aalpha_{\sfA_{\J}}$ and $\bbeta_{\sfA_{\J}}$
be defined as
\begin{align}
\alpha_{\sfA_{\J}}
&\equiv
\prod_{j\in\J}\alpha_{\sfA_j}
\label{eq:multi-alpha-def}
\\
\beta_{\sfA_{\J}}
&\equiv
\prod_{j\in\J}\lrB{\beta_{\sfA_j}+1}-1.
\label{eq:multi-beta-def}
\end{align}
Then
\begin{align}
p_{\sfA_{\K}}\lrsb{\lrb{
    A_{\K}: \lrB{\G\setminus\{\uu_{\K}\}}\cap\C_{A_{\K}}(A_{\K}\uu_{\K})\neq \emptyset
}}
&
\leq 
\sum_{\substack{
    \J\subset\K\\
    \J\neq\emptyset
}}
\frac{
  \lrbar{\G_{\J|\J^c}}
  \alpha_{\sfA_{\J}}\lrB{\beta_{\sfA_{\J^c}}+1}
}
{\prod_{j\in\J}\lrbar{\im\A_j}
}
+\beta_{\sfA_{\K}}
\label{eq:multi-CRP}
\end{align}
for all $\G\subset\lrB{\U_{\K}}^n$ and $\uu_{\K}\in\lrB{\U_{\K}}^n$,
and
\begin{align}
p_{A_{\K}\sfaa_{\K}}\lrsb{\lrb{
    (A_{\K},\ba_{\K}):
    \T\cap\C_{A_{\K}}(\ba_{\K})=\emptyset
}}
&
\leq
\alpha_{\sfA_{\K}}-1
+
\sum_{\J\subsetneq\K}
\frac{
  \lrB{\prod_{j\in\J^c}\lrbar{\im\A_j}}
  \lrbar{\T_{\J|\J^c}}
  \alpha_{\sfA_{\J}}\lrB{\beta_{\sfA_{\J^c}}+1}
}
{|\T|}
\label{eq:multi-SP}
\end{align}
for all $\T\subset\lrB{\U_{\K}}^n$.
Furthermore, if $(\aalpha_{\sfA_j},\bbeta_{\sfA_j})$ satisfies
(\ref{eq:alpha}) and (\ref{eq:beta}) for all $j\in\K$, then
\begin{align}
\limn \alpha_{\sfA_{\J}}(n)=1
\label{eq:multi-alpha}
\\
\limn \beta_{\sfA_{\J}}(n)=0
\label{eq:multi-beta}
\end{align}
for every $\J\subset\K$.
\end{lem}

\begin{lem}
\label{lem:multi-lemma-bound}
If $\T$ is a subset of $\T_{U_{\K},\gamma}$, then
\begin{equation*}
\max_{\uu_{\J^c}\in\T_{\U_{\J^c}}}
\lrbar{\T_{\U_{\J}|\U_{\J^c}}\lrsb{\uu_{\J^c}}}
\leq
2^{n\lrB{
    H(U_{\J}|U_{\J^c})
    +\eta_{\U_{\J}\U_{\J^c}}(\gamma|\gamma)
}}
\label{eq:multi-lemma-bound}
\end{equation*}
for every non-empty set $\J\subsetneq\K$.
\end{lem}

\subsection{Construction of Minimum-Divergence Operation}
\label{sec:md-multi}

In this section, we introduce the construction
of the minimum-divergence operation
\begin{equation*}
\hg_{A_{\K}}\lrsb{\ba_{\K}}
\equiv\arg\min_{\uu_{\K}\in\C_{A_{\K}}(\ba_{\K})}
D\lrsb{\nu_{\uu_{\K}}\|\mu_{U_{\K}}}
\label{eq:md-multi}
\end{equation*} 
by assuming that $\U_j\equiv\{0,1\}$ is a finite field,
and $A_j$ is an  $l_{\A_j}\times n$ matrix on $\{0,1\}$ for all
$j\in\K$.
It is a extension of the operation introduced in Section~\ref{sec:binary}.

Let $t_{\uu_{\K}}(b^k)\in\{0,\ldots,n\}$ be the number of occurrences of
$b^k\in\{0,1\}^k$ in $\uu_{\K}\in\{0,1\}^{kn}$.
We have
\begin{align*}
\hg_{A_{\K}}\lrsb{\ba_{\K}}
&=
\arg\min_{\substack{
    \uu_{\K,\bt}:
    \\
    \bt\in\cH
}}
D\lrsb{\nu_{\uu_{\K,\bt}}\|\mu_{U_{\K}}},
\end{align*}
where
\begin{align}
\uu_{\K,\bt}
&\equiv
\begin{cases}
  \uu_{\K}
  &\text{if}\ \exists \uu_{\K}\ \text{s.t.}\ \bt_{\uu_{\K}}=\bt,
  A_j\uu_j=\ba_j\ \text{for all}\ j\in\K
  \\
  \text{`error'}
  &\text{otherwise}
\end{cases} 
\label{eq:ut-multi}
\\
\cH
&\equiv\lrb{\bt:
  t\lrsb{b^k}\in\{0,\ldots,n\},
  \sum_{b^k}t\lrsb{b^k}=n
}
\notag
\end{align}
and $D(\nu_{\uu_{\K,\bt}}\|\mu_{U_{\K}})\equiv\infty$
when the operation (\ref{eq:ut-multi}) returns an error.
It should be noted that
the cardinality of the set $\cH$
is at most $[n+1]^{2^k}$.

We can apply the linear programming technique introduced in \cite{FWK05}
to the construction of (\ref{eq:ut-multi}).
By using the technique introduced in \cite{FWK05},
the conditions $A_j\uu_j=\ba_j$ for all $j\in\K$
can be replaced by linear inequalities.
In the following, we focus on the condition $\bt_{\uu_{\K}}=\bt$
for a given $\bt\in\cH$.

First, for each $b^k\in\{0,1\}^k$, we introduce the auxiliary variable
$\lrsb{s_1\lrsb{b^k},\ldots,s_n\lrsb{b^k}}$
defined as
\begin{equation}
s_i\lrsb{b^k}\equiv
\begin{cases}
  1,&\text{if}\ u_{j,i}=b_j\ \text{for all}\ j\in\K
  \\
  0,&\text{otherwise}.
\end{cases}
\label{eq:si-multi}
\end{equation}
Then the condition $\bt_{\uu_{\K}}=\bt$ is equivalent to
the linear equality
\[
\sum_{i=1}^n s_i\lrsb{b^k}=t\lrsb{b^k}
\]
for all $b^k\in\{0,1\}^k$.

Next, we show that the relation (\ref{eq:si-multi}) can be replaced by
linear inequalities.
Let $\cS\lrsb{b^k}\subset\{0,1\}^{k+1}$ be defined as
\begin{equation}
\cS(b^k)
\equiv
\lrb{\lrsb{1,b^k}}
\cup\lrB{\bigcup_{b^{\prime k}\neq b^k}\lrb{\lrsb{0,b^{\prime k}}}}.
\label{eq:Sb}
\end{equation}

We have the following lemma,
which is proved in Appendix~\ref{sec:proof-multi-binary}.
\begin{lem}
\label{lem:multi-binary}
The convex hull of $\cS\lrsb{b^k}$ is a polytope represented by the
following inequalities:
\begin{align*}
v_0&\geq 0
\\
v_j&\geq 0\ \text{for $j\in\K$ s.t. $b_j=0$}
\\
v_j&\leq 1\ \text{for $j\in\K$ s.t. $b_j=1$}
\\
v_0+[-1]^{b_j}v_j&\leq 1-b_j
\\
v_0+\sum_{j\in\K}[-1]^{b_j}v_j&\geq 1-\sum_{j\in\K}b_j
\end{align*}
where constants $b^k$ and variables $v_0,\ldots,v_k$ are considered to
be real numbers.
Furthermore, there is no non-integral vertex of this polytope.
\end{lem}

By using this lemma, we have $[2[k+1]n+1]2^k$ inequalities,
which replace the condition $\bt_{\uu_{\K}}=\bt$, defined as
\begin{align*}
\sum_{i=1}^ns_i\lrsb{b^k}&=t\lrsb{b^k}
\\
s_i\lrsb{b^k}&\geq 0
\\
u_{j,i}&\geq 0\ \text{for $j\in\K$ s.t. $b_j=0$}
\\
u_{j,i}&\leq 1\ \text{for $j\in\K$ s.t. $b_j=1$}
\\
s_i\lrsb{b^k}+[-1]^{b_j}u_{j,i}&\leq 1-b_j
\ \text{for all $j\in\K$}
\\
s_i\lrsb{b^k}+\sum_{j\in\K}[-1]^{b_j}u_{j,i}
&\geq 1-\sum_{j\in\K}b_j,
\end{align*}
where we take all $i\in\{1,2,\ldots,n\}$ and $b^k\in\{0,1\}^k$.

\section{Proof of Theorems}

\subsection{Proof of Theorem \ref{thm:hash-linear}}
\label{sec:proof-linear}

For a type $\bt$, let $\C_{\bt}$ be defined as
\begin{gather*}
\C_{\bt} \equiv \lrb{\uu\in\U^n :\ \bt(\uu)=\bt}.
\end{gather*}
We assume that $\pA\lrsb{\lrb{A: A\uu=\zero}}$ depends on $\uu$ only
through the type $\bt(\uu)$.
For a given $\uu\in\C_{\bt}$, we define 
\begin{align*}
p_{\sfA,\bt}
&\equiv \pA\lrsb{\lrb{A: A\uu=\zero}}.
\end{align*}
We use the following lemma.
\begin{lem}[{\cite[Lemma 9]{HASH}}]
Let $(\alphaA,\betaA)$ be defined by (\ref{eq:alpha-linear})
and (\ref{eq:beta-linear}). Then
\begin{align}
\begin{split}
  \alphaA
  &=
  |\im\A|\max_{\bt\in \hcH_{\sfA}}p_{\sfA,\bt} 
\end{split}
\label{eq:alpha2}
\\
\betaA
&=
\sum_{\bt\in\cH\setminus\hcH_{\sfA}}|\C_{\bt}|p_{\sfA,\bt},
\label{eq:beta2}
\end{align}
where $\cH$ is a set of all types of length $n$ except the type
of the zero vector.
\end{lem}

Now we prove Theorem~\ref{thm:hash-linear}.
It is enough to show (\ref{eq:hash})
because (\ref{eq:alpha}), (\ref{eq:beta}) are satisfied from the
assumption of the theorem.
Since function $A$ is linear, we have
\begin{align}
p_{\sfA}(\{A: A\uu = A\uu'\})
&=p_{\sfA}(\{A: A[\uu-\uu']=\zero\})
\notag
\\
&=p_{\sfA,\bt(\uu-\uu')}
\end{align}
Then, for $\uu\neq\uu'$ satisfying $\bt(\uu-\uu')\in\hcH_{\sfA}$, we have
\begin{align}
p_{\sfA}(\{A: A\uu = A\uu'\})
&=p_{\sfA,\bt(\uu-\uu')}
\notag
\\
&\leq
\max_{\bt\in \hcH_{\sfA}}p_{\sfA,\bt} 
\notag
\\
&=
\frac{\alphaA}{|\im\A|},
\end{align}
where the last inequality comes from (\ref{eq:alpha2}).
Then we have the fact that
$p_{\sfA}(\{A: A\uu = A\uu'\})>\alphaA/|\im\A|$ implies
$\bt(\uu-\uu')\notin\hcH_{\sfA}$.
Finally, we have
\begin{align}
\sum_{\substack{
    \uu'\in\U^n\setminus\{\uu\}
    \\
    p_{\sfA}(\{A: A\uu = A\uu'\})>\frac{\alphaA}{|\im\A|}
}}
p_{\sfA}\lrsb{\lrb{A: A\uu = A\uu'}}
&\leq
\sum_{\substack{
    \uu'\in\U^n\setminus\{\uu\}
    \\
    \bt(\uu-\uu')\in\cH\setminus\hcH_{\sfA}
}}
p_{\sfA,\bt(\uu-\uu')}
\notag
\\
&\leq
\sum_{\bt\in\cH\setminus\hcH_{\sfA}}
\sum_{\substack{
    \uu'\in\U^n\setminus\{\uu\}
    \\
    \bt(\uu-\uu')=\bt
}}
p_{\sfA,\bt}
\notag
\\
&\leq
\sum_{\bt\in\cH\setminus\hcH_{\sfA}}
|\C_{\bt}|p_{\sfA,\bt}
\notag
\\
&=\betaA,
\end{align}
where the equality comes from (\ref{eq:beta2}).
\hfill\QED

\subsection{Proof of Theorem \ref{thm:sw}}
\label{sec:proof-sw}
Let $(\xx,\yy)$ be the output of correlated sources.
We define
\begin{align}
&\bullet (\xx,\yy)\notin\T_{XY,\gamma}
\tag{SW1}
\\
&\bullet \exists (\xx',\yy')\neq(\xx,\yy)\ \text{s.t.}
\ \xx'\in\C_A(A\xx),\ \yy'\in\C_B(B\yy),
D(\nu_{\xx'\yy'}\|\mu_{XY})\leq D(\nu_{\xx\yy}\|\mu_{XY}).
\tag{SW2}
\end{align}
It should be noted that (SW2) includes
the following three cases:
\begin{itemize}
\item there is $(\xx',\yy')$
satisfying $\xx'\neq\xx$, $\yy'=\yy$,
$\xx'\in\C_A(A\xx)$,  and
$D(\nu_{\xx'\yy'}\|\mu_{XY})\leq D(\nu_{\xx\yy}\|\mu_{XY})$,
\item there is $(\xx',\yy')$
satisfying $\xx'=\xx$, $\yy'\neq\yy$,
$\yy'\in\C_B(B\yy)$, and
$D(\nu_{\xx'\yy'}\|\mu_{XY})\leq D(\nu_{\xx\yy}\|\mu_{XY})$,
\item there is $(\xx',\yy')$ satisfying
$\xx'\neq\xx$, $\yy'\neq\yy$,
$\xx'\in\C_A(A\xx)$, $\yy'\in\C_B(B\yy)$, and
$D(\nu_{\xx'\yy'}\|\mu_{XY})\leq D(\nu_{\xx\yy}\|\mu_{XY})$.
\end{itemize}
Since a decoding error occurs when
at least one of the conditions (SW1) and (SW2) is satisfied,
the error probability is upper bounded by
\begin{align}
\Error_{XY}(A,B)
&\leq
\mu_{XY}(\E_1)+\mu_{XY}(\E_1^c\cap\E_2),
\label{eq:sw0}
\end{align}
where we define
\[
\E_i\equiv\{(\xx,\yy): \text{(SW$i$)}\}.
\]

First, we evaluate $E_{AB}\lrB{\mu_{XY}(\E_1)}$.
From Lemma \ref{lem:typical-prob}, we have
\begin{equation}
E_{AB}\lrB{\mu_{XY}(\E_1)}
\leq
\frac{\delta}2
\label{eq:sw1}
\end{equation}
for all sufficiently large $n$.

Next, we evaluate $E_{AB}\lrB{\mu_{XY}(\E_1^c\cap\E_2)}$.
When (SW2) is satisfied but (SW1) is not,
we have
\begin{align*}
&\lrB{\T_{XY,\gamma}\setminus\{(\xx,\yy)\}}\cap\lrB{\C_{A}(A\xx)\times\C_B(B\yy)}
\neq\emptyset.
\end{align*}
Applying
Lemma~\ref{lem:hash-AxB-CRP}
by letting
$\G\equiv\T_{XY,\gamma}$,
we have
\begin{align}
&
E_{\sfA\sfB}\lrB{\mu_{XY}(\E_1^c\cap\E_2)}
\notag
\\*
&=
\sum_{(\xx,\yy)\in\T_{XY,\gamma}}\mu_{XY}(\xx,\yy)
p_{\sfA\sfB}\lrsb{\lrb{
    (A,B): \text{(SW2)}
}}
\notag
\\
&\leq
\sum_{(\xx,\yy)\in\T_{XY,\gamma}}\mu_{XY}(\xx,\yy)
p_{\sfA\sfB}\lrsb{\lrb{
    (A,B):
    \lrB{\T_{XY,\gamma}\setminus\{(\xx,\yy)\}}\cap\lrB{\C_{A}(A\xx)\times\C_B(B\yy)}
    \neq\emptyset
}}
\notag
\\
&\leq
\frac{|\G|\alphaA\alphaB}{|\im\A||\im\B|}
+
\frac{\lrB{\disp\max_{\yy\in\G_{\Y}}|\G_{\X|\Y}(\yy)|}\alphaA[\betaB+1]}
{|\im\A|}
+\frac{\lrB{\disp\max_{\xx\in\G_{\X}}|\G_{\Y|\X}(\xx)|}\alphaB[\betaA+1]}
{|\im\B|}
+\betaA+\betaB+\betaA\betaB
\notag
\\
&\leq
\frac{|\T_{XY,\gamma}|\alphaA\alphaB}{|\im\A||\im\B|}
+
\frac{\lrB{\disp\max_{\yy\in\T_{Y,\gamma}}|\T_{X|Y,\gamma}(\yy)|}\alphaA[\betaB+1]}
{|\im\A|}
+\frac{\lrB{\disp\max_{\xx\in\T_{X,\gamma}}|\T_{Y|X,\gamma}(\xx)|}\alphaB[\betaA+1]}
{|\im\B|}
+\betaA+\betaB+\betaA\betaB
\notag
\\
\begin{split}
  &\leq
  \frac{2^{n[H(X,Y)+\eta_{\X\Y}(\gamma)]}\alphaA\alphaB}{|\im\A||\im\B|}
  +
  \frac{2^{n[H(X|Y)+\eta_{\X|\Y}(\gamma|\gamma)]}\alphaA[\betaB+1]}
  {|\im\A|}
  +
  \frac{2^{n[H(Y|X)+\eta_{\Y|\X}(\gamma|\gamma)]}\alphaB[\betaA+1]}
  {|\im\B|}
  \\*
  &\quad
  +\betaA+\betaB+\betaA\betaB
\end{split} 
\notag
\\
\begin{split}
  &=
  2^{-n[R_X+R_Y-H(X,Y)-\eta_{\X\Y}(\gamma)]}\alphaA\alphaB
  +
  2^{-n[R_X-H(X|Y)-\eta_{\X|\Y}(\gamma|\gamma)]}\alphaA[\betaB+1]
  \\*
  &\quad
  +
  2^{-n[R_Y-H(Y|X)-\eta_{\Y|\X}(\gamma|\gamma)]}\alphaB[\betaA+1]
  +\betaA+\betaB+\betaA\betaB
\end{split} 
\notag
\\
&\leq
\frac{\delta}4
\label{eq:sw2}
\end{align}
for all sufficiently large $n$ by taking an appropriate $\gamma>0$,
where the third inequality comes from the fact that
$\G_{\X}\subset\T_{X,\gamma}$, $\G_{\Y}\subset\T_{Y,\gamma}$,
$\G_{\X|\Y}(\yy)\subset\T_{X|Y,\gamma}(\yy)$,
and $\G_{\Y|\X}(\xx)\subset\T_{Y|X,\gamma}(\xx)$,
which are obtained from Lemma~\ref{lem:typical-trans}.
The fourth inequality comes from Lemma~\ref{lem:typical-number} and
the last inequality comes from (\ref{eq:swx})--(\ref{eq:swxy})
and conditions (\ref{eq:alpha}), (\ref{eq:beta}) of ensembles
$(\bcA,\bpA)$ and $(\bcB,\bpB)$.

Finally, from (\ref{eq:sw0})--(\ref{eq:sw2}),
for all $\delta>0$ and for all sufficiently large $n$ there are $A$ and
$B$ such that
\[
\Error_{XY}(A,B)<\delta.
\]
\hfill\QED

\begin{rem}
We need the conditions $\xx'\in\T_{X,\gamma}$ and 
$\yy'\in\T_{Y,\gamma}$ in the maximum-likelihood decoder (\ref{eq:ML})
because
$\G\equiv\lrb{(\xx',\yy'): \mu_{XY}(\xx',\yy')\geq \mu_{XY}(\xx,\yy)}$
does not satisfy
\begin{align*}
\max_{\yy\in\G_{\Y}}|\G_{\X|\Y}(\yy)|
&\leq 2^{n[H(X|Y)+\eta_{\X|\Y}(\gamma|\gamma)]}
\\
\max_{\xx\in\G_{\X}}|\G_{\Y|\X}(\xx)|
&\leq 2^{n[H(Y|X)+\eta_{\Y|\X}(\gamma|\gamma)]}
\end{align*}
in general.
It is unknown if we can remove these conditions.
It seems that the same problem appears in \cite[Eq.\ (20)]{SWLDPC}.
\end{rem}

\subsection{Proof of Theorem \ref{thm:broadcast}}
In the following, we assume that ensembles $(\bcA,\bpA)$,
$(\bcA',\bpAp)$, $(\bcB,\bpB)$, and $(\bcB',\bpBp)$
have a strong hash property.
Then, from Lemma \ref{lem:hash-AB}, ensembles $(\bhcA,\bphA)$ and
$(\bhcB,\bphB)$, defined by
\begin{align*}
\hA\uu&\equiv(A\uu,A'\uu)
\\
\hB\vv&\equiv(B\vv,B'\vv)
\end{align*}
have an $(\aalphahA,\bbetahA)$-hash property and an
$(\aalphahB,\bbetahB)$-hash property, respectively, where
\begin{align*}
|\im\hcA|&=|\im\A||\im\A'|
\\
|\im\hcB|&=|\im\B||\im\B'|
\\
\alphahA&\equiv \alphaA\alphaAp
\\
\betahA&\equiv \betaA+\betaAp
\\
\alphahB&\equiv \alphaB\alphaBp
\\
\betahB&\equiv \betaB+\betaBp.
\end{align*}

For $\bbetaA$ and $\bbetaB$ satisfying
\begin{align*}
&\limn\betaA(n)=0
\\
&\limn\betaB(n)=0,
\end{align*}
let $\kkappa\equiv\{\kappa(n)\}_{n=1}^{\infty}$ be a sequence satisfying
\begin{gather}
\limn\kappa(n)=\infty
\label{eq:k1}
\\
\limn \kappa(n)\betaA(n)=0
\label{eq:k2a}
\\
\limn \kappa(n)\betaB(n)=0
\label{eq:k2b}
\\
\limn\frac{\log\kappa(n)}n=0.
\label{eq:k3}
\end{gather}
For example, we obtain such a $\kkappa$ by letting
\begin{equation*}
\kappa(n)
\equiv
\begin{cases}
  n^{\xi}
  &\text{if}\ \beta(n)=o\lrsb{n^{-\xi}}, \xi>0
  \\
  \frac 1{\sqrt{\beta(n)}},
  &\text{otherwise}
\end{cases}
\end{equation*}
for every $n$, where $\beta(n)\equiv\max\{\betaA(n),\betaB(n)\}$.
If $\beta(n)$ is not $o\lrsb{n^{-\xi}}$,
there is $\kappa'>0$ such that
$\beta(n)n^{\xi}>\kappa'$
and
\begin{align}
\frac{\log\kappa(n)}n
\notag
&=
\frac{\log\frac 1{\beta(n)}}{2n}
\notag
\\
&\leq \frac{\log\frac{n^{\xi}}{\kappa'}}{2n}
\notag
\\
&=\frac{\xi\log n-\log\kappa'}{2n}
\end{align}
for all sufficiently large $n$.
This implies that $\kkappa$ satisfies (\ref{eq:k3}).
In the following, $\kappa$ denotes $\kappa(n)$.

From (\ref{eq:rY})--(\ref{eq:rYRYrZRZ-lower}), and (\ref{eq:k3}),
we have the fact that there is $\gamma>0$ such that
\begin{gather}
\e\geq \eta_{\U\V}(\gamma)+\frac{\log\kappa}n
\\
r_Y> H(U|Y)+\zeta_{\U|\Y}(2\gamma|2\gamma)
\label{eq:eAzeta}
\\
r_Z> H(V|Z)+\zeta_{\V|\Z}(2\gamma|2\gamma)
\label{eq:eBzeta}
\\
r_Y+R_Y\leq H(U)-\e-\eta_{\V|\U}(\gamma|\gamma)
\label{eq:RYeta}
\\
r_Z+R_Z\leq H(V)-\e-\eta_{\U|\V}(\gamma|\gamma)
\label{eq:RZeta}
\\
r_Y+R_Y+r_Z+R_Z
\leq
H(U,V)-\eta_{\U\V}(\gamma)-\frac{\log\kappa}n
\label{eq:ekappa}
\\
r_Y+R_Y+r_Z+R_Z
\geq
H(U,V)-\e
\label{eq:RYRZeAB}
\end{gather}
for all sufficiently large $n$.

We have
\begin{align}
|\T_{UV,\gamma}|
&\geq 2^{n[H(U,V)-\eta_{\U\V}(\gamma)]}
\notag
\\
&\geq
\kappa2^{n[r_Y+R_Y+r_Z+R_Z]}
\notag
\\
&=
\kappa|\im\A||\im\A'||\im\B||\im\B'|
\notag
\\
&=
\kappa|\im\hcA||\im\hcB|
\end{align}
for all sufficiently large $n$,
where
the first inequality comes from Lemma \ref{lem:typical-number},
and the second inequality comes from (\ref{eq:ekappa}).
This implies that
there is $\T\subset\T_{UV,\gamma}$ such that
\begin{align}
\kappa
&\leq
\frac{|\T|}{|\im\hcA||\im\hcB|}
\leq 
2\kappa
\label{eq:T}
\end{align}
for all sufficiently large $n$,
where we construct such $\T$ by taking
$|\T|$ elements from $\T_{UV,\gamma}$
in ascending order of divergence.

Now we prove the theorem.
Let $\mm$ and $\ww$ be messages.
Let $\uu$ and $\vv$ denotes the $\U^n$-component
(resp. $\V^n$-component) of $g_{AA'BB'}(\ba,\mm,\bb,\ww)$, that is,
\begin{align*}
(\uu,\vv)&\equiv \hg_{AA'BB'}(\ba,\mm,\bb,\ww).
\end{align*}
Let $\yy$ and $\zz$ be the channel outputs.
We define
\begin{align}
&\bullet (\uu,\vv)\in\T\subset\T_{UV,\gamma}
\tag{BC1}
\label{eq:bc1}
\\
&\bullet
(\yy,\zz)\in
\T_{YZ|UVX,\gamma}(\uu,\vv,\xx)
\tag{BC2}
\label{eq:bc2}
\\
&\bullet g_{A}(\ba|\yy)=\uu
\tag{BC3}
\label{eq:bc3}
\\
&\bullet g_{B}(\bb|\zz)=\vv,
\tag{BC4}
\label{eq:bc4}
\end{align}
where we define $\xx\equiv\ff(\uu,\vv)$.
Then the error probability (\ref{eq:error}) is upper bounded by
\begin{align}
&
\Error_{YZ|X}(A,A',B,B',\ba,\bb)
\notag
\\*
\begin{split}
  &\leq
  p_{\sfmm\sfww YZ}(\cS_1^c)
  +p_{\sfmm\sfww YZ}(\cS_2^c)
  +p_{\sfmm\sfww YZ}(\cS_1\cap\cS_2\cap\cS_3^c)
  +p_{\sfmm\sfww YZ}(\cS_1\cap\cS_2\cap\cS_4^c),
\end{split}
\label{eq:error0}
\end{align}
where
\begin{align*}
\cS_i
&\equiv
\lrb{
  (\mm,\ww,\yy,\zz): \text{(BC$i$)}
}.
\end{align*}
In the following, we define
\begin{align*}
\haa&\equiv (\ba,\mm)
\\
\hbb&\equiv (\bb,\ww)
\end{align*}
and
\begin{gather*}
\C_{\hA\hB}(\haa,\hbb)
\equiv
\lrB{\C_{A}(\ba)\cap\C_{A'}(\mm)}\times\lrB{\C_{B}(\bb)\cap\C_{B'}(\ww)}.
\end{gather*}

First, we evaluate
$E_{\sfA\sfA'\sfB\sfB'\sfaa\sfbb}\lrB{p_{\sfmm\sfww YZ}(\cS_1^c)}$.
From Lemma~\ref{lem:hash-AxB-SP}, we have
\begin{align}
&
E_{\sfA\sfA'\sfB\sfB'\sfaa\sfbb}\lrB{p_{\sfmm\sfww YZ}(\cS_1^c)}
\notag
\\*
&=
p_{\sfA\sfA'\sfB\sfB'\sfaa\sfbb\sfmm\sfww}
\lrsb{\lrb{
    (A,A',B,B',\ba,\bb,\mm,\ww):
    \hg_{AA'BB'}(\ba,\mm,\bb,\ww)\notin\T
}}
\notag
\\
&\leq
p_{\sfhA\sfhB\sfhaa\sfhbb}
\lrsb{\lrb{
    (\hA,\hB,\haa,\hbb):
    \T\cap\C_{\hA\hB}(\haa,\hbb)=\emptyset
}}
\notag
\\
\begin{split}
  &\leq
  \alphahA\alphahB-1
  +
  \frac{|\im\hcB|\lrB{\disp\max_{\vv\in\T_{\V}}|\T_{\U|\V}(\vv)|}\alphahA[\betahB+1]}
  {|\T|}
  +
  \frac{|\im\hcA|\lrB{\disp\max_{\xx\in\T_{\U}}|\T_{\V|\U}(\uu)|}\alphahB[\betahA+1]}
  {|\T|}
  \\*
  &\quad
  +\frac{|\im\hcA||\im\hcB|\lrB{\betahA+\betahB+\betahA\betahB+1}}
  {|\T|}.
\end{split}
\label{eq:error1-saturation}
\end{align} 
From Lemma~\ref{lem:typical-trans} and the fact that $\T\subset\T_{XY,\gamma}$,
we have the fact that $\uu\in\T_{\U}$ implies $\uu\in\T_{U,\gamma}$
and $\vv\in\T_{\V|\U}(\uu)$ implies $\vv\in\T_{V|U,\gamma}(\uu)$.
Furthermore, from Lemma~\ref{lem:typical-number}, we have
\begin{align}
\max_{\vv\in\T_{\V}}|\T_{\U|\V}(\vv)|
&\leq 
\max_{\vv\in\T_{V,\gamma}}|\T_{U|V,\gamma}(\vv)|
\notag
\\
&\leq
2^{n[H(U|V)+\eta_{\U|\V}(\gamma|\gamma)]}.
\label{eq:maxT-bc}
\end{align}
Then we have
\begin{align}
\frac{|\im\hcB|\disp\max_{\vv\in\T_{\V}}|\T_{\U|\V}(\vv)|}
{|\T|}
&\leq
\frac{2^{n[H(U|V)+\eta_{\U|\V}(\gamma|\gamma)]}}
{\kappa|\im\hcA|}
\notag
\\
&=
\frac{2^{-n[r_Y+R_Y-H(U|V)-\eta_{\U|\V}(\gamma|\gamma)]}}
{\kappa}
\notag
\\
&\leq
\frac 1{\kappa}
\label{eq:limit-bc1}
\end{align}
where the first inequality comes from (\ref{eq:T}), (\ref{eq:maxT-bc}),
and the second inequality comes from the fact that
\begin{align}
r_Y+R_Y
&=
r_Y+R_Y+r_Z+R_Z-r_Z-R_Z
\notag
\\
&\geq
H(U,V)-\e-r_Z-R_Z
\notag
\\
&\geq
H(U,V)-\e-H(V)+\e+\eta_{\U|\V}(\gamma|\gamma)
\notag
\\
&=
H(U|V)+\eta_{\U|\V}(\gamma|\gamma)
\end{align}
which is obtained from (\ref{eq:RYeta}) and (\ref{eq:RYRZeAB}).
Similarly,
we have
\begin{align}
\frac{|\im\hcA|\disp\max_{\uu\in\T_{\U}}|\T_{\V|\U}(\uu)|}
{|\T|}
&
\leq
\frac{2^{-n[r_Z+R_Z-H(V|U)-\eta_{\V|\U}(\gamma|\gamma)]}}
{\kappa}
\notag
\\
&\leq \frac 1{\kappa}
\label{eq:limit-bc2}
\end{align}
Furthermore, we have
\begin{align}
\frac{|\im\hcA||\im\hcB|}
{|\T|}
&\leq
\frac1{\kappa}.
\label{eq:limit-bc3}
\end{align} 
from (\ref{eq:T}).
Then, from (\ref{eq:k1}),
(\ref{eq:error1-saturation}), (\ref{eq:limit-bc1})
(\ref{eq:limit-bc2}), (\ref{eq:limit-bc3}),
and the properties (\ref{eq:alpha}) and (\ref{eq:beta})
for ensembles $(\bhcA,\bphA)$ and $(\bhcB,\bphB)$,
we have
\begin{align}
E_{\sfA\sfA'\sfB\sfB'\sfaa\sfbb}\lrB{p_{\sfmm\sfww YZ}(\cS_1^c)}
&\leq
\frac {\delta}4
\label{eq:error1}
\end{align}
for all $\delta>0$ and sufficiently large $n$.

Next, we evaluate
$E_{\sfA\sfA'\sfB\sfB'\sfaa\sfbb}\lrB{p_{\sfmm\sfww YZ}(\cS_2^c)}$.
We have
\begin{align}
E_{\sfA\sfA'\sfB\sfB'\sfaa\sfbb}\lrB{p_{\sfmm\sfww YZ}(\cS_2^c)}
&=
E_{\sfhA\sfhB\sfhaa\sfhbb}\lrB{
  \mu_{YZ|X}\lrsb{
    \lrB{\T_{YZ|UVX,\gamma}(\UU,\VV,\XX)}^c|\XX
  }
}
\notag
\\
&=
E_{\sfhA\sfhB\sfhaa\sfhbb}\lrB{
  \mu_{YZ|UVX}\lrsb{
    \lrB{\T_{YZ|UVX,\gamma}(\UU,\VV,\XX)}^c|\UU,\VV,\XX
  }
}
\notag
\\
&\leq 2^{-n[\gamma-\lambda_{\U\V\X\Y\Z}]}
\notag
\\
&\leq
\frac {\delta}4
\label{eq:error2}
\end{align}
for all $\delta>0$ and sufficiently large $n$,
where
we define $\XX\equiv\ff(\UU,\VV)$,
the second equality comes from (\ref{eq:markov-broadcast}),
and the first inequality comes from Lemma \ref{lem:typical-prob}.

Next, we evaluate
$E_{\sfA\sfA'\sfB\sfB'\sfaa\sfbb}\lrB{p_{\sfmm\sfww YZ}(\cS_1\cap\cS_2\cap\cS_3^c)}$
and
$E_{\sfA\sfA'\sfB\sfB'\sfaa\sfbb}\lrB{p_{\sfmm\sfww YZ}(\cS_1\cap\cS_2\cap\cS_4^c)}$.
In the following, we assume (\ref{eq:bc1}), (\ref{eq:bc2}), and
$g_{A}(\ba|\yy)\neq\uu$,
where the last assumption is equivalent to $g_{A}(A\uu|\yy)\neq\uu$
from (\ref{eq:bc1}).
From (\ref{eq:UVfUV}) and Lemma~\ref{lem:typical-trans},
we have the fact that $(\uu,\vv)\in\T$ implies
$(\uu,\vv,\xx)\in\T_{UVX,\gamma}$.
From Lemma \ref{lem:typical-trans}, we have
$(\uu,\yy)\in\T_{UY,2\gamma}$ and $\uu\in\T_{U|Y,2\gamma}(\yy)$.
Then there is $\uu'\in\C_A(A\uu)$ such that $\uu'\neq\uu$ and
\begin{align}
\mu_{U|Y}(\uu'|\yy)
\notag
&\geq
\mu_{U|Y}(\uu|\yy)
\notag
\\
&\geq
2^{-n[H(U|Y)+\zeta_{\U|\Y}(2\gamma|2\gamma)]},
\end{align}
where the second inequality comes from Lemma \ref{lem:typical-prob}.
This implies that $[\G(\yy)\setminus\{\uu\}]\cap\C_A(A\uu)\neq\emptyset$,
where
\[
\G(\yy)\equiv\lrb{
  \uu': \mu_{U|Y}(\uu'|\yy)\geq 2^{-[H(U|Y)+\zeta_{\U|\Y}(2\gamma|2\gamma)]}
}.
\]
From Lemma~\ref{lem:collision}, we have 
\begin{align}
E_{\sfA}\lrB{
  \chi(g_{\sfA}(\sfA\uu|\yy)\neq \uu)
}
&\leq
p_{\sfA}\lrsb{\lrb{
    A:
    \lrB{\G(\yy)\setminus\{\uu\}}\cap\C_A(A\uu)\neq\emptyset
}}
\notag
\\
&\leq
\frac{|\G(\yy)|\alphaA}
{|\im\A|}
+\betaA
\notag
\\
&\leq
2^{-n[r_Y-H(U|Y)-\zeta_{\U|\Y}(2\gamma|2\gamma)]}\alphaA
+\betaA.
\label{eq:proof-bc-collision}
\end{align}
the last inequality comes from the definition of $r_Y$ and the fact that
\[
|\G(\yy)|\leq 2^{n[H(U|Y)+\zeta_{\U|\Y}(2\gamma|2\gamma)]}.
\]
Then we have
\begin{align}
&
E_{\sfA\sfA'\sfB\sfB'\sfaa\sfbb}\lrB{
  p_{\sfmm\sfww YZ}(\cS_1\cap\cS_2\cap\cS_3^c)
}
\notag
\\*
&\leq
E_{\sfA\sfA'\sfB\sfB'\sfaa\sfbb\sfmm\sfww}\lrB{
  \sum_{(\uu,\vv)\in\T}
  \chi(\hg_{\sfA\sfA'\sfB\sfB'}(\sfaa,\sfmm,\sfbb,\sfww)=(\uu,\vv))
  \sum_{\yy\in\T_{Y|UVX,\gamma}(\uu,\vv,\xx)}
  \mu_{Y|X}(\yy|\xx)\chi(g_{\sfA}(\sfaa|\yy)\neq \uu)
}
\notag
\\
&\leq
E_{\sfA\sfA'\sfB\sfB'\sfaa\sfbb\sfmm\sfww}\left[
  \sum_{(\uu,\vv)\in\T}
  \chi(\sfA\uu=\sfaa)\chi(\sfA'\uu=\sfmm)
  \chi(\sfB\vv=\sfbb)\chi(\sfB'\vv=\sfww)
  \sum_{\yy\in\T_{Y|UVX,\gamma}(\uu,\vv,\xx)}
  \mu_{Y|X}(\yy|\xx)
  \chi(g_{\sfA}(\sfaa|\yy)\neq \uu)
\right]
\notag
\\
&=
\sum_{(\uu,\vv)\in\T}
\sum_{\yy\in\T_{Y|UVX,\gamma}(\uu,\vv,\xx)}
\mu_{Y|X}(\yy|\xx)
E_{\sfA}
\lrB{
  \chi(g_{\sfA}(\sfA\uu|\yy)\neq \uu)
  E_{\sfA'\sfB\sfB'\sfaa\sfbb\sfmm\sfww}
  \lrB{
    \chi(\sfA\uu=\sfaa)
    \chi(\sfA'\uu=\sfmm)
    \chi(\sfB\vv=\sfbb)
    \chi(\sfB'\vv=\sfww)
  }
}
\notag
\\
&=
\frac 1{|\im\A||\im\A'||\im\B||\im\B'|}
\sum_{(\uu,\vv)\in\T}
\sum_{\yy\in\T_{Y|UVX,\gamma}(\uu,\vv,\xx)}
\mu_{Y|X}(\yy|\xx)
E_{\sfA}\lrB{
  \chi(g_{\sfA}(\sfA\uu|\yy)\neq \uu)
}
\notag
\\
&\leq
\lrB{
  2^{-n[r_Y-H(U|Y)-\zeta_{\U|\Y}(2\gamma|2\gamma)]}\alphaA
  +\betaA
}
\sum_{(\uu,\vv)\in\T}
\frac 1{|\im\A||\im\A'||\im\B||\im\B'|}
\notag
\\
&\leq
2\kappa
\lrB{
  2^{-n[r_Y-H(U|Y)-\zeta_{\U|\Y}(2\gamma|2\gamma)]}\alphaA
  +\betaA
}
\notag
\\
&\leq
\frac{\delta}4
\label{eq:error3}
\end{align}
for all $\delta>0$ and sufficiently large $n$,
where the second equality comes from Lemma~\ref{lem:E} that appears
in Appendix~\ref{sec:lemE},
the fourth inequality comes from (\ref{eq:T}),
and the last inequality comes from (\ref{eq:k2a}),
(\ref{eq:eAzeta}),
and the conditions (\ref{eq:alpha}), (\ref{eq:beta}) of $(\bcA,\bpA)$.
Similarly, we have
\begin{align}
E_{\sfA\sfA'\sfB\sfB'\sfaa\sfbb}\lrB{
  p_{\sfmm\sfww YZ}(\cS_1\cap\cS_2\cap\cS_4^c)
}
&\leq
2\kappa
\lrB{
  2^{-n[r_Z-H(V|Z)-\zeta_{\V|\Z}(2\gamma|2\gamma)]}\alphaB
  +\betaB
}
\notag
\\
&\leq
\frac{\delta}4
\label{eq:error4}
\end{align}
for all $\delta>0$ and sufficiently large $n$.

Finally, from (\ref{eq:error0}), (\ref{eq:error1}), (\ref{eq:error2}),
(\ref{eq:error3}),
and (\ref{eq:error4}),
we have the fact that
for all $\delta>0$ and sufficiently large $n$
there are $A\in\A$, $A'\in\A'$, $B\in\B$, $B'\in\B'$,
$\ba\in\im\A$, and $\bb\in\im\B$
satisfying (\ref{eq:error}).
\hfill\QED

\section{Conclusion}
The constructions of the Slepian-Wolf source code and the broadcast
channel code were presented.
The proof of the theorems is based on the notion of a strong hash
property for an ensemble of functions, where two lemmas called
`collision-resistance property' and `saturation property'
introduced~\cite{HASH} are extended from a single domain to multiple
domains.
Since an ensemble of sparse matrices has a strong hash property,
we can construct codes by using sparse matrices and it is expected that
we can use the efficient approximation algorithms for encoding/decoding.
It should be noted that the capacity region for the general broadcast
channel coding is unknown and we hope that our approach give us a hint
for deriving the general capacity region.

\appendix

\subsection{Basic Property of Ensemble}
\label{sec:lemE}

We review the following lemma, which is proved in \cite{HASH}.

\begin{lem}[{\cite[Lemma 9]{HASH}}]
\label{lem:E}
Assume that random variables $\sfA$ and $\sfaa$ are independent.
Then,
\begin{align*}
E_{\sfaa}\lrB{\chi(A\uu=\sfaa)}
&=\frac 1{|\im\A|}
\end{align*}
for any $A$ and $\uu$
and
\begin{align*}
E_{\sfA\sfaa}\lrB{\chi(\sfA\uu=\sfaa)}
&=\frac 1{|\im\A|}
\end{align*}
for any $\uu$.
\end{lem}

\subsection{Proof of Lemma \ref{lem:whash}}
\label{sec:wash}
If an ensemble satisfies (H4), then we have
\begin{align}
&
\sum_{\substack{
    \uu\in\T
    \\
    \uu'\in\T'
}}
\pA\lrsb{\lrb{A: A\uu = A\uu'}}
\notag
\\*
&=
\sum_{\uu\in\T\cap\T'}
\pA\lrsb{\lrb{A: A\uu = A\uu'}}
+
\sum_{\uu\in\T}
\sum_{\substack{
    \uu'\in\T'\setminus\{\uu\}
    \\
    \pA(\{A: A\uu = A\uu'\})\leq\frac{\alphaA}{|\im\A|}
}}
\pA\lrsb{\lrb{A: A\uu = A\uu'}}
\notag
\\*
&\quad
+
\sum_{\uu\in\T}
\sum_{\substack{
    \uu'\in\T'\setminus\{\uu\}
    \\
    \pA(\{A: A\uu = A\uu'\})>\frac{\alphaA}{|\im\A|}
}}
\pA\lrsb{\lrb{A: A\uu = A\uu'}}
\notag
\\
&\leq
|\T\cap\T'|
+
\sum_{\uu\in\T}
\sum_{\substack{
    \uu'\in\T'\setminus\{\uu\}
    \\
    \pA(\{A: A\uu = A\uu'\})\leq\frac{\alphaA}{|\im\A|}
}}
\frac{\alphaA}{|\im\A|}
+
\sum_{\uu\in\T}
\betaA
\notag
\\
&\leq
|\T\cap\T'|
+
\frac{|\T||\T'|\alphaA}{|\im\A|}
+
|\T|\betaA
\notag
\\
&\leq
|\T\cap\T'|
+
\frac{|\T||\T'|\alphaA}{|\im\A|}
+
\min\{|\T|,|\T'|\}\betaA
\end{align} 
for any $\T$ and $\T'$ satisfying $|\T|\leq |\T'|$.
\hfill\QED

\subsection{Proof of Lemma \ref{lem:hash-AB}}
\label{sec:hash-AB}
Let
\begin{align*}
p_{\sfA,\uu,\uu'}&\equiv \pA(\{A: A\uu=A\uu'\})
\\
p_{\sfA',\uu,\uu'}&\equiv \pAp(\{A': A'\uu=A'\uu'\}).
\end{align*}
Then we have
\begin{align}
&
\sum_{\substack{
    \uu'\in\U^n\setminus\{\uu\}
    \\
    p_{\sfhA,\uu,\uu'}>\frac{\alpha_{\sfhA}}{|\im\hcA|}
}}
p_{\sfhA}(\{\hA: \hA\uu=\hA\uu'\})
\notag
\\*
&\leq
\sum_{\substack{
    \uu'\in\U^n\setminus\{\uu\}
    \\
    p_{\sfA,\uu,\uu'}p_{\sfA',\uu,\uu'}
    >\frac{\alphaA\alphaAp}{|\im\A||\im\A'|}
}}
p_{\sfA,\uu,\uu'}p_{\sfA',\uu,\uu'}
\notag
\\
&=
\sum_{\substack{
    \uu'\in\U^n\setminus\{\uu\}
    \\
    p_{\sfA,\uu,\uu'}p_{\sfA',\uu,\uu'}
    >\frac{\alphaA\alphaAp}{|\im\A||\im\A'|}
    \\
    p_{\sfA,\uu,\uu'}>\frac{\alphaA}{|\im\A|}
}}
p_{\sfA,\uu,\uu'}p_{\sfA',\uu,\uu'}
+
\sum_{\substack{
    \uu'\in\U^n\setminus\{\uu\}
    \\
    p_{\sfA,\uu,\uu'}p_{\sfA',\uu,\uu'}
    >\frac{\alphaA\alphaAp}{|\im\A||\im\A'|}
    \\
    p_{\sfA,\uu,\uu'}\leq\frac{\alphaA}{|\im\A|}
}}
p_{\sfA,\uu,\uu'}p_{\sfA',\uu,\uu'}
\notag
\\
&\leq
\sum_{\substack{
    \uu'\in\U^n\setminus\{\uu\}
    \\
    p_{\sfA,\uu,\uu'}>\frac{\alphaA}{|\im\A|}
}}
p_{\sfA,\uu,\uu'}p_{\sfA',\uu,\uu'}
+
\sum_{\substack{
    \uu'\in\U^n\setminus\{\uu\}
    \\
    p_{\sfA',\uu,\uu'}>\frac{\alphaAp}{|\im\A'|}
}}
p_{\sfA,\uu,\uu'}p_{\sfA',\uu,\uu'}
\notag
\\
&\leq
\sum_{\substack{
    \uu'\in\U^n\setminus\{\uu\}
    \\
    p_{\sfA,\uu,\uu'}>\frac{\alphaA}{|\im\A|}
}}
p_{\sfA,\uu,\uu'}
+
\sum_{\substack{
    \uu'\in\U^n\setminus\{\uu\}
    \\
    p_{\sfA',\uu,\uu'}>\frac{\alphaAp}{|\im\A'|}
}}
p_{\sfA',\uu,\uu'}
\notag
\\
&=
\betaA+\betaAp
\notag
\\
&=
\beta_{\sfhA},
\label{eq:proof-AoplusB}
\end{align}
where the first inequality comes from the fact that
$\im\hcA\subset\im\A\times\im\A'$, the first equality comes from the
fact that $\sfA$ and $\sfA'$ are mutually independent and the last
inequality comes from the fact that $p_{\sfA,\uu,\uu'}\leq 1$ and
$p_{\sfA',\uu,\uu'}\leq 1$.
Since $(\aalpha_{\sfhA},\bbeta_{\sfhA})$ satisfies (\ref{eq:alpha}) and
(\ref{eq:beta}), then we have the fact that $(\bhcA,\bp_{\sfhA})$ has an
$(\aalpha_{\sfhA},\bbeta_{\sfhA})$-hash property.
\hfill\QED

\subsection{Proof of Lemmas \ref{lem:hash-AxB-CRP} and \ref{lem:hash-AxB-SP}}
\label{sec:hash-AxB}
Let
\begin{align*}
&p_{\sfA,\uu,\uu'}\equiv \pA(\{A: A\uu=A\uu'\})
\\
&p_{\sfB,\vv,\vv'}\equiv \pB(\{A: B\vv=A\vv'\}).
\end{align*}
For $(\uu',\vv')\in\T'$, we have
\begin{align}
\sum_{\substack{
    (\uu,\vv)\in\T
    \\
    p_{\sfA,\uu,\uu'}\leq\frac{\alphaA}{|\im\A|}
    \\
    p_{\sfB,\vv,\vv'}>\frac{\alphaB}{|\im\B|}
}}
p_{\sfA,\uu,\uu'}
p_{\sfB,\vv,\vv'}
&=
\sum_{\substack{
    \vv\in\T_{\V}
    \\
    p_{\sfB,\vv,\vv'}>\frac{\alphaB}{|\im\B|}
}}
p_{\sfB,\vv,\vv'}
\sum_{\substack{
    \uu\in\T_{\U|\V}(\vv)
    \\
    p_{\sfA,\uu,\uu'}\leq\frac{\alphaA}{|\im\A|}
}}
p_{\sfA,\uu,\uu'}
\notag
\\
&\leq
\sum_{\substack{
    \vv\in\T_{\V}
    \\
    p_{\sfB,\vv,\vv'}>\frac{\alphaB}{|\im\B|}
}}
p_{\sfB,\vv,\vv'}
\sum_{\substack{
    \uu\in\T_{\U|\V}(\vv)
    \\
    p_{\sfA,\uu,\uu'}\leq\frac{\alphaA}{|\im\A|}
}}
\frac{\alphaA}{|\im\A|}
\notag
\\
&\leq
\frac{\lrB{\disp\max_{\vv\in\T_{\V}}|\T_{\U|\V}(\vv)|}\alphaA}{|\im\A|}
\sum_{\substack{
    \vv\in\T_{\V}
    \\
    p_{\sfB,\vv,\vv'}>\frac{\alphaB}{|\im\B|}
}}
p_{\sfB,\vv,\vv'}
\notag
\\
&
\leq
\frac{\lrB{\disp\max_{\vv\in\T_{\V}}|\T_{\U|\V}(\vv)|}\alphaA}{|\im\A|}
\lrB{
  \sum_{\substack{
      \vv\in\V^n\setminus\{\vv'\}
      \\
      p_{\sfB,\vv,\vv'}>\frac{\alphaB}{|\im\B|}
  }}
  p_{\sfB,\vv,\vv'}
  +p_{\sfB,\vv,\vv}
}
\notag
\\
&\leq
\frac{\lrB{\disp\max_{\vv\in\T_{\V}}|\T_{\U|\V}(\vv)|}\alphaA[\betaB+1]}
{|\im\A|}.
\label{eq:TUV}
\end{align}
Similarly, we have
\begin{align}
\begin{split}
  &
  \sum_{\substack{
      (\uu,\vv)\in\T
      \\
      p_{\sfA,\uu,\uu'}>\frac{\alphaA}{|\im\A|}
      \\
      p_{\sfB,\vv,\vv'}\leq\frac{\alphaB}{|\im\B|}
  }}
  p_{\sfA,\uu,\uu'}
  p_{\sfB,\vv,\vv'}
  \leq
  \frac{\lrB{\disp\max_{\uu\in\T_{\U}}|\T_{\V|\U}(\uu)|}\alphaB[\betaA+1]}
  {|\im\B|}.
\end{split} 
\label{eq:TVU}
\end{align}
for $(\uu',\vv')\in\T'$.
We also have
\begin{align}
&
\sum_{\substack{
    (\uu,\vv)\in\T
    \\
    p_{\sfA,\uu,\uu'}>\frac{\alphaA}{|\im\A|}
    \\
    p_{\sfB,\vv,\vv'}>\frac{\alphaB}{|\im\B|}
}}
p_{\sfA,\uu,\uu'}
p_{\sfB,\vv,\vv'}
\notag
\\*
&\leq
\lrbar{\T\cap\{(\uu',\vv')\}}
+
\sum_{\substack{
    \uu\in\T_{\U}\setminus\{\uu'\}
    \\
    p_{\sfA,\uu,\uu'}>\frac{\alphaA}{|\im\A|}
}}
p_{\sfA,\uu,\uu'}p_{\sfB,\vv,\vv}
+
\sum_{\substack{
    \vv\in\T_{\V}\setminus\{\vv'\}
    \\
    p_{\sfB,\vv,\vv'}>\frac{\alphaB}{|\im\B|}
}}
p_{\sfA,\uu,\uu}p_{\sfB,\vv,\vv'}
+
\sum_{\substack{
    (\uu,\vv)\in\T
    \\
    \uu\neq\uu',\vv\neq\vv'
    \\
    p_{\sfA,\uu,\uu'}>\frac{\alphaA}{|\im\A|}
    \\
    p_{\sfB,\vv,\vv'}>\frac{\alphaB}{|\im\B|}
}}
p_{\sfA,\uu,\uu'}p_{\sfB,\vv,\vv'}
\notag
\\
&\leq
\lrbar{\T\cap\{(\uu',\vv')\}}
+
\sum_{\substack{
    \uu\in\U^n\setminus\{\uu'\}
    \\
    p_{\sfA,\uu,\uu'}>\frac{\alphaA}{|\im\A|}
}}
p_{\sfA,\uu,\uu'}
+
\sum_{\substack{
    \vv\in\V^n\setminus\{\vv'\}
    \\
    p_{\sfB,\vv,\vv'}>\frac{\alphaB}{|\im\B|}
}}
p_{\sfB,\vv,\vv'}
+
\sum_{\substack{
    \uu\in\U^n\setminus\{\uu'\}
    \\
    p_{\sfA,\uu,\uu'}>\frac{\alphaA}{|\im\A|}
}}
p_{\sfA,\uu,\uu'}
\sum_{\substack{
    \vv\in\V^n\setminus\{\vv'\}
    \\
    p_{\sfB,\vv,\vv'}>\frac{\alphaB}{|\im\B|}
}}
p_{\sfB,\vv,\vv'}
\notag
\\
&\leq
\lrbar{\T\cap\{(\uu',\vv')\}}+
\betaA+\betaB+\betaA\betaB
\label{eq:Tsmall}
\end{align}
for $(\uu',\vv')\in\T'$.
Finally, we have
\begin{align}
&
\sum_{\substack{
    (\uu,\vv)\in\T
    \\
    (\uu',\vv')\in\T'
}}
p_{\sfA\sfB}\lrsb{\lrb{(A,B): A\uu=A\uu',B\vv=B\vv'}}
\notag
\\*
&
=
\sum_{(\uu',\vv')\in\T'}
\sum_{\substack{
    (\uu,\vv)\in\T
    \\
    p_{\sfA,\uu,\uu'}\leq\frac{\alphaA}{|\im\A|}
    \\
    p_{\sfB,\vv,\vv'}\leq\frac{\alphaB}{|\im\B|}
}}
p_{\sfA,\uu,\uu'}
p_{\sfB,\vv,\vv'}
+
\sum_{(\uu',\vv')\in\T'}
\sum_{\substack{
    (\uu,\vv)\in\T
    \\
    p_{\sfA,\uu,\uu'}\leq\frac{\alphaA}{|\im\A|}
    \\
    p_{\sfB,\vv,\vv'}>\frac{\alphaB}{|\im\B|}
}}
p_{\sfA,\uu,\uu'}
p_{\sfB,\vv,\vv'}
\notag
\\*
&\quad
+
\sum_{(\uu',\vv')\in\T'}
\sum_{\substack{
    (\uu,\vv)\in\T
    \\
    p_{\sfA,\uu,\uu'}>\frac{\alphaA}{|\im\A|}
    \\
    p_{\sfB,\vv,\vv'}\leq\frac{\alphaB}{|\im\B|}
}}
p_{\sfA,\uu,\uu'}
p_{\sfB,\vv,\vv'}
+
\sum_{(\uu',\vv')\in\T'}
\sum_{\substack{
    (\uu,\vv)\in\T
    \\
    p_{\sfA,\uu,\uu'}>\frac{\alphaA}{|\im\A|}
    \\
    p_{\sfB,\vv,\vv'}>\frac{\alphaB}{|\im\B|}
}}
p_{\sfA,\uu,\uu'}
p_{\sfB,\vv,\vv'}
\notag
\\
&
\leq
\sum_{(\uu',\vv')\in\T'}
\frac{|\T|\alphaA\alphaB}{|\im\A||\im\B|}
+
\sum_{(\uu',\vv')\in\T'}
\frac{\lrB{\disp\max_{\vv\in\T_{\V}}|\T_{\U|\V}(\vv)|}\alphaA[\betaB+1]}
{|\im\A|}
+
\sum_{(\uu',\vv')\in\T'}
\frac{\lrB{\disp\max_{\uu\in\T_{\U}}|\T_{\V|\U}(\uu)|}\alphaB[\betaA+1]}
{|\im\B|}
\notag
\\*
&\quad
+\sum_{(\uu',\vv')\in\T'}
\lrB{\lrbar{\T\cap\{(\uu',\vv')\}}
  +\beta_A+\beta_B+\beta_A\beta_B}
\notag
\\
\begin{split}
  &
  \leq
  \lrbar{\T\cap\T'}
  +\frac{|\T||\T'|\alphaA\alphaB}{|\im\A||\im\B|}
  +
  \frac{|\T'|\lrB{\disp\max_{\vv\in\T_{\V}}|\T_{\U|\V}(\vv)|}\alphaA[\betaB+1]}
  {|\im\A|}
  +
  \frac{|\T'|\lrB{\disp\max_{\uu\in\T_{\U}}|\T_{\V|\U}(\uu)|}\alphaB[\betaA+1]}
  {|\im\B|}
  \\*
  &\quad
  +|\T'|\lrB{\betaA+\betaB+\betaA\betaB},
\end{split}
\end{align}
where the first inequality comes from
(\ref{eq:TUV})--(\ref{eq:Tsmall}).
This implies that
the joint ensemble $(\bcA\times\bcB,\bpAB)$
satisfies (\ref{eq:whash}) by letting
\begin{align*}
\alpha
&\equiv\alphaA\alphaB
\\
\beta
&\equiv
\frac{\lrB{\disp\max_{\vv\in\T_{\V}}|\T_{\U|\V}(\vv)|}\alphaA[\betaB+1]}
{|\im\A|}
+
\frac{\lrB{\disp\max_{\uu\in\T_{\U}}|\T_{\V|\U}(\uu)|}
  \alphaB[\betaA+1]}
{|\im\B|}
+\betaA+\betaB+\betaA\betaB.
\end{align*} 
Then we have Lemmas \ref{lem:hash-AxB-CRP} and \ref{lem:hash-AxB-SP}
from Lemmas~\ref{lem:collision} and~\ref{lem:saturation}, respectively.
\hfill\QED

\subsection{Proof of Lemma \ref{lem:multi-lemma}}
\label{sec:proof-multi}

It is easy to show (\ref{eq:multi-alpha}) and (\ref{eq:multi-beta})
from the properties (\ref{eq:alpha}) and (\ref{eq:beta})
of $(\aalpha_{\sfA_j},\bbeta_{\sfA_j})$ for all $j\in\K$.
In the following, we show (\ref{eq:multi-CRP}) and (\ref{eq:multi-SP}).
Let $p_{\uu_j,\uu'_j}$ be defined as
\begin{align*}
p_{\uu_j,\uu'_j}
&\equiv
p_{\sfA_j}\lrsb{\lrb{
    A_j:
    A_j\uu_j=A_j\uu'_j
}}.
\end{align*}

First, we have
\begin{align}
&
\sum_{\substack{
    \uu_{\K}\in\T
    \\
    p_{\uu_j,\uu'_j}
    \leq\frac{\alpha_{\sfA_j}}{|\im\A_j|}
    \ \text{for all $j\in\J$}
    \\
    p_{\uu_j,\uu'_j}
    >\frac{\alpha_{\sfA_j}}{|\im\A_j|}
    \ \text{for all $j\in\J^c$}
    \\
}}
\prod_{j\in\K} p_{\uu_j,\uu'_j}
\notag
\\*
&=
\sum_{\substack{
    \uu_{\J^c}\in\T_{\U_{\J^c}}
    \\
    p_{\uu_j,\uu'_j}
    >\frac{\alpha_{\sfA_j}}{|\im\A_j|}
}}
\prod_{j\in\J^c} p_{\uu_j,\uu'_j}
\sum_{\substack{
    \uu_{\J}\in\T_{\U_{\J}|\U_{\J^c}}\lrsb{\uu_{\J^c}}
    \\
    p_{\uu_j,\uu'_j}
    \leq\frac{\alpha_{\sfA_j}}{|\im\A_j|}
}}
\prod_{j\in\J}p_{\uu_j,\uu'_j}
\notag
\\
&\leq
\lrB{\disp
  \max_{\uu_{\J^c}\in\T_{\U_{\J^c}}}
  \lrbar{\T_{\U_{\J}|\U_{\J^c}}\lrsb{\uu_{\J^c}}}
}
\lrB{\prod_{j\in\J}\frac{\alpha_{\sfA_j}}{|\im\A_j|}}
\sum_{\substack{
    \uu_{\J^c}\in\T_{\U_{\J^c}}
    \\
    p_{\uu_j,\uu'_j}
    >\frac{\alpha_{\sfA_j}}{|\im\A_j|}
}}
\prod_{j\in\J^c} p_{\uu_j,\uu'_j}
\notag
\\
&\leq
\lrB{\disp
  \max_{\uu_{\J^c}\in\T_{\U_{\J^c}}}
  \lrbar{\T_{\U_{\J}|\U_{\J^c}}\lrsb{\uu_{\J^c}}}
}
\lrB{\prod_{j\in\J}\frac{\alpha_{\sfA_j}}{|\im\A_j|}}
\prod_{j\in\J^c}
\lrB{
  \sum_{\substack{
      \uu_j\in\U_j\setminus\lrb{\uu'_j}
      \\
      p_{\uu_j,\uu'_j}
      >\frac{\alpha_{\sfA_j}}{|\im\A_j|}
  }}
  p_{\uu_j,\uu'_j}
  +
  p_{\uu'_j,\uu'_j}
}
\notag
\\
&\leq
\lrB{\disp
  \max_{\uu_{\J^c}\in\T_{\U_{\J^c}}}
  \lrbar{\T_{\U_{\J}|\U_{\J^c}}\lrsb{\uu_{\J^c}}}
}
\lrB{\prod_{j\in\J}\frac{\alpha_{\sfA_j}}{|\im\A_j|}}
\prod_{j\in\J^c}
\lrB{\beta_{\sfA_j}+1}
\notag
\\
&=
\frac{
  \lrbar{\T_{\J|\J^c}}
  \alpha_{\sfA_{\J}}\lrB{\beta_{\sfA_{\J^c}}+1}
}
{\prod_{j\in\J}|\im\A_j|}
\label{eq:lemma-multi1}
\end{align}
for $\uu'_{\K}\in\T'$ and a non-empty set $\J\subsetneq\K$,
where the second inequality comes from the property (\ref{eq:hash}) of
$(\A_j,p_{\sfA_j})$ for all $j\in\K$ and the fact that
$p_{\uu'_j,\uu'_j}=1$, and the last equality comes from (\ref{eq:maxT}).

Next, we have
\begin{align}
\sum_{\substack{
    \uu_{\K}\in\T\setminus\lrb{\uu'_{\K}}
    \\
    p_{\uu_j,\uu'_j}
    >\frac{\alpha_{\sfA_j}}{|\im\A_j|}
    \ \text{for all $j\in\K$}
}}
\prod_{j\in\K} p_{\uu_j,\uu'_j}
&\leq
\sum_{\substack{
    \uu_{\K}\in\U_{\K}
    \\
    p_{\uu_j,\uu'_j}
    >\frac{\alpha_{\sfA_j}}{|\im\A_j|}
    \ \text{for all $j\in\K$}
}}
\prod_{j\in\K} p_{\uu_j,\uu'_j}
-
\prod_{j\in\K} p_{\uu'_j,\uu'_j}
\notag
\\
&=
\prod_{j\in\J}
\sum_{\substack{
    \uu_j\in\U_j
    \\
    p_{\uu_j,\uu'_j}
    >\frac{\alpha_{\sfA_j}}{|\im\A_j|}
}}
p_{\uu_j,\uu'_j}
-1
\notag
\\
&=
\prod_{j\in\J}
\lrB{
  \sum_{\substack{
      \uu_j\in\U_j\setminus\{\uu'_j\}
      \\
      p_{\uu_j,\uu'_j}
      >\frac{\alpha_{\sfA_j}}{|\im\A_j|}
  }}
  p_{\uu_j,\uu'_j}
  +
  p_{\uu'_j,\uu'_j}
}
-1
\notag
\\
&=
\prod_{j\in\K}\lrB{\beta_{\sfA_j}+1}-1
\notag
\\
&
=
\frac{|\T_{\emptyset|\K}|\alpha_{\sfA_{\emptyset}}\lrB{\beta_{\sfA_{\K}}+1}}
{\prod_{j\in\emptyset}|\im\A_j|}
-1
\label{eq:lemma-multi2}
\end{align}
for $\uu'_{\K}\in\T'$, where we use 
the fact that $p_{\uu'_j,\uu'_j}=1$
in the second inequality,
and the last inequality comes from
(\ref{eq:maxT})--(\ref{eq:multi-beta-def}).

Finally, we have
\begin{align}
&\sum_{\substack{
    \uu_{\K}\in\T
    \\
    \uu'_{\K}\in\T'
}}
p_{\sfA_{\K}}\lrsb{\lrb{
    A_{\K}:
    A_j\uu_j=A_j\uu'_j
    \ \text{for all}\ j\in\K
}}
\notag
\\*
&=
\sum_{\uu'_{\K}\in\T'}
\sum_{\uu_{\K}\in\T}
\prod_{j=1}^k p_{\uu_j,\uu'_j}
\notag
\\
&\leq
|\T\cap\T'|
+
\sum_{\uu'_{\K}\in\T'}
\sum_{\substack{
    \uu_{\K}\in\T
    \\
    p_{\uu_j,\uu'_j}
    \leq\frac{\alpha_{\sfA_j}}{|\im\A_j|}
    \ \text{for all $j\in\K$}
}}
\prod_{j\in\K} p_{\uu_j,\uu'_j}
+
\sum_{\uu'_{\K}\in\T'}
\sum_{\substack{
    \J\subsetneq\K
    \\
    \J\neq\emptyset
}}
\sum_{\substack{
    \uu_{\K}\in\T
    \\
    p_{\uu_j,\uu'_j}
    \leq\frac{\alpha_{\sfA_j}}{|\im\A_j|}
    \ \text{for all $j\in\J$}
    \\
    p_{\uu_j,\uu'_j}
    >\frac{\alpha_{\sfA_j}}{|\im\A_j|}
    \ \text{for all $j\in\J^c$}
}}
\prod_{j\in\K} p_{\uu_j,\uu'_j}
\notag
\\*
&\quad
+
\sum_{\uu'_{\K}\in\T'}
\sum_{\substack{
    \uu_{\K}\in\T\setminus\lrb{\uu'_{\K}}
    \\
    p_{\uu_j,\uu'_j}
    >\frac{\alpha_{\sfA_j}}{|\im\A_j|}
    \ \text{for all}\ j\in\K
}}
\prod_{j\in\K} p_{\uu_j,\uu'_j}
\notag
\\
\begin{split}
  &\leq
  |\T\cap\T'|
  +
  \sum_{\uu'_{\K}\in\T'}
  \sum_{\substack{
      \uu_{\K}\in\T
      \\
      p_{\uu_j,\uu'_j}
      \leq \frac{\alpha_{\sfA_j}}{|\im\A_j|}
      \ \text{for all}\ j\in\K
  }}
  \prod_{j\in\K}
  \frac{\alpha_{\sfA_j}}{|\im\A_j|}
  +
  \sum_{\uu'_{\K}\in\T'}
  \sum_{\substack{
      \J\subsetneq\K
      \\
      \J\neq\emptyset
  }}
  \frac{\lrbar{\T_{\J|\J^c}}\alpha_{\sfA_{\J}}\lrB{\beta_{\sfA_{\J^c}}+1}
  }
  {\prod_{j\in\J}|\im\A_j|}
  \\*
  &\quad
  +
  \sum_{\uu'_{\K}\in\T'}
  \lrB{
    \frac{|\T_{\emptyset|\K}|\alpha_{\sfA_{\emptyset}}\lrB{\beta_{\sfA_{\K}}+1}}
    {\prod_{j\in\emptyset}\im\A_j}
    -1
  }
\end{split}
\notag
\\
&
\leq
|\T\cap\T'|
+
\frac{|\T||\T'|\alpha_{\sfA_{\K}}}
{|\im\A_{\K}|}
+
|\T'|
\lrB{
  \sum_{\J\subsetneq\K}
  \frac{\lrbar{\T_{\J|\J^c}}\alpha_{\sfA_{\J}}\lrB{\beta_{\sfA_{\J^c}}+1}
  }
  {\prod_{j\in\J}|\im\A_j|}
  -1
}
\label{eq:lemma-multi3}
\end{align}
for all $\T,\T'\subset\times_{j=1}^k\U_j$,
where the first inequality comes from 
the fact that
\[
p_{\uu'_j,\uu'_j}=1
>\frac{\alpha_{\sfA_j}}{|\im\A_j|}
\]
for all $j\in\K$,
the second inequality comes from 
(\ref{eq:lemma-multi1}) and (\ref{eq:lemma-multi2}).
This implies that
the joint ensemble $(\bcA_{\K},\bp_{\sfA_{\K}})$
satisfies (\ref{eq:whash}) by letting
\begin{align*}
\alpha
&\equiv
\alpha_{\sfA_{\K}}
\\
\beta
&\equiv
\sum_{\J\subsetneq\K}
\frac{\lrbar{\T_{\J|\J^c}}\alpha_{\sfA_{\J}}\lrB{\beta_{\sfA_{\J^c}}+1}
}
{\prod_{j\in\J}|\im\A_j|}
-1
\end{align*} 
Then we have (\ref{eq:multi-CRP}) and (\ref{eq:multi-SP})
from Lemmas~\ref{lem:collision} and~\ref{lem:saturation},
respectively, where we use the relation
\[
\beta_{\sfA_{\K}}=
\frac{
  |\G_{\emptyset|\K}|\alpha_{\sfA_{\emptyset}}
  \lrB{\beta_{\sfA_{\K}}+1}
}
{\prod_{j\in\emptyset}|\im\A_j|}
-1
\]
for the proof of (\ref{eq:multi-CRP}).
\hfill\QED

\subsection{Proof of Lemma \ref{lem:multi-lemma-bound}}
\label{sec:proof-multi-bound}
Since $\T$ is a subset of $\T_{U_{\K},\gamma}$,
we have
\begin{align*}
\T_{\U_{\J^c}}&\subset\T_{U_{\J^c},\gamma}
\\
\T_{\U_{\J}|\U_{\J^c}}(\uu_{\J^c})
&\subset \T_{U_{\J}\left|U_{\J^c}\right.,\gamma}(\uu_{\J^c})
\end{align*}
from Lemma~\ref{lem:typical-trans}.
From Lemma~\ref{lem:typical-number}, we have
\begin{align*}
\max_{\uu_{\J^c}\in\T_{\U_{\J^c}}}
\lrbar{\T_{\U_{\J}|\U_{\J^c}}\lrsb{\uu_{\J^c}}}
&\leq
\max_{\uu_{\J^c}\in\T_{U_{\J^c},\gamma}}
\lrbar{\T_{U_{\J}|U_{\J^c},\gamma}\lrsb{\uu_{\J^c}}}
\notag
\\
&\leq
2^{n\lrB{H(U_{\J}|U_{\J^c})+\eta_{\U_{\J}|\U_{\J^c}}(\gamma|\gamma)}}.
\end{align*}
\hfill\QED

\subsection{Proof of Lemma \ref{lem:multi-binary}}
\label{sec:proof-multi-binary}

Before proving the lemma, we prepare 
some definitions for a convex polytope (see~\cite{Z06}).
Let $\cP\subset\Real^k$ be a convex polytope.
Then the linear inequality $\cc\vv\leq c_0$ is {\em valid} for $\cP$
if this inequality is satisfied for all $\vv\in\cP$.
The {\em face} of the polytope $\cP$ is a set defined by
\[
\cP\cap\{\vv\in\Real^k:\cc\vv=c_0\},
\]
where $\cc\vv\leq c_0$ is valid for $\cP$.
If the dimension of a face is $0$, it is called a {\em vertex}.
\begin{lem}[{\cite[Proposition 2.3(ii)]{Z06}}]
\label{lem:face}
Every intersection of the faces of $\cP$ is  a face of $\cP$.
\end{lem}
\begin{lem}[{\cite[p.~52]{Z06}}]
\label{lem:vertex-iff}
A point $\vv\in\cP$ cannot be expressed as a convex combination
of $\cP\setminus\{\vv\}$ if and only if it is a vertex of $\cP$.
\end{lem}

For $k\geq 1$,
let $\cP\lrsb{b^k}\subset\Real^{k+1}$ be a convex polytope defined by
inequalities
\begin{align}
v_0&\geq 0
\label{eq:poly-v0}
\\
v_j&\geq 0
\quad\text{for $j\in\K$ s.t. $b_j=0$}
\label{eq:poly-v0b0}
\\
v_j&\leq 1
\quad\text{for $j\in\K$ s.t. $b_j=1$}
\label{eq:poly-v1b1}
\\
v_0+[-1]^{b_j}v_j&\leq 1-b_j
\quad\text{for all $j\in\K$}
\label{eq:poly-vb}
\\
v_0
+
\sum_{j\in\K}[-1]^{b_j}v_j
&\geq
1-
\sum_{j\in\K}
b_j
\label{eq:poly-v0eq1}
\end{align}
where constants $b^k\in\{0,1\}^k$ and variables $v_0,\ldots,v_k$
are considered to be real numbers.

Then we have the following lemmas.
\begin{lem}
\label{lem:valid}
Let $\vv\equiv(v_0,v_1,\ldots,v_k)$.
Then, for all $j\in\{0\}\cup\K$, inequalities
\begin{align}
-v_j&\leq 0
\label{eq:poly-vj0}
\\
v_j&\leq 1
\label{eq:poly-vj1}
\end{align}
are valid for $\cP(b^k)$.
\end{lem} 
\begin{proof}
First, we show (\ref{eq:poly-vj0}) and (\ref{eq:poly-vj1})
for  $j=0$.
The inequality $-v_0\leq 0$ is valid
because it is 
equivalent to (\ref{eq:poly-v0}).
If $b^{(1)}=0$ then we have
\[
v_0\leq v_0+v_1\leq 1
\]
from (\ref{eq:poly-v0b0}) and (\ref{eq:poly-vb}).
On the other hand, if $b_j=1$ then we have
\[
v_0\leq v_1\leq 1
\]
from (\ref{eq:poly-v1b1}) and (\ref{eq:poly-vb}).
Hence we have the fact that $v_0\leq 1$ is valid for $\cP(b^k)$.

Next, we show (\ref{eq:poly-vj0}) and (\ref{eq:poly-vj1})
for  $j\in\K$ satisfying $b_j=0$.
The inequality (\ref{eq:poly-vj0})
is valid because it is equivalent to (\ref{eq:poly-v0b0}).
The validity of (\ref{eq:poly-vj1}) is obtained by
\[
v_j\leq v_0+v_j\leq 1,
\]
which comes from (\ref{eq:poly-v0}) and (\ref{eq:poly-vb}).

Finally, we show (\ref{eq:poly-vj0}) and (\ref{eq:poly-vj1})
for  $j\in\K$ satisfying $b_j=1$.
The inequality (\ref{eq:poly-vj1})
is valid because it is equivalent to (\ref{eq:poly-v1b1}).
The validity of (\ref{eq:poly-vj1}) is obtained by
\[
v_j\geq v_0\geq 0,
\]
which comes from (\ref{eq:poly-v0}) and (\ref{eq:poly-vb}).
\end{proof}

\begin{lem}
\label{lem:i-vertex}
For all $b^k\in\{0,1\}^k$,
an integral point $\vv$ in $\cP(b^k)$ satisfies
$\vv\in\{0,1\}^{k+1}$.
\end{lem}
\begin{proof}
From Lemma~\ref{lem:valid},
the integral point $\vv=(v_0,v_1,\ldots,v_k)$ 
satisfies either $v_j=0$ or $v_j=1$ 
for all $j\in\{0\}\cup\K$.
\end{proof}

\begin{lem}
\label{lem:vertex}
For all $b^k\in\{0,1\}^k$,
$(1,b^k)$ is a member of $\cP(b^k)$
and $(0,u^k)$ is a member $\cP(b^k)$
for all $u^k\in\{0,1\}^k$ satisfying $u^k\neq b^k$.
\end{lem}
\begin{proof}
The first statement $(1,b^k)\in\cP(b^k)$ is proved by showing
that $\vv=(1,b^k)$ satisfies
inequalities (\ref{eq:poly-v0})--(\ref{eq:poly-v0eq1})
for all $b^k\in\{0,1\}^k$.
It is clear that $\vv=(1,b^k)$ satisfies
(\ref{eq:poly-v0})--(\ref{eq:poly-v1b1}).
Inequalities (\ref{eq:poly-vb}) and (\ref{eq:poly-v0eq1})
come from the fact that
$v_0=1$ and $b^k$ satisfies
\begin{equation}
[-1]^{b_j}b_j=-b_j
\label{eq:bk}
\end{equation}
for all $j\in\K$.

The second statement $(0,u^k)\in\cP(b^k)$
is proved by showing
that $\vv=(0,u^k)$ satisfies
inequalities (\ref{eq:poly-v0})--(\ref{eq:poly-v0eq1})
for all $b^k\in\{0,1\}^k$ and $u^k\in\{0,1\}^k$ satisfying
$u^k\neq b^k$.
It is clear that $\vv=(0,u^k)$ satisfies
(\ref{eq:poly-v0})--(\ref{eq:poly-v1b1})
for all $u^k\in\{0,1\}^k$.
Since
\begin{equation}
[-1]^{b_j}u_j
=
\begin{cases}
  1-b_j, &\text{if}\ u_j\neq b_j
  \\
  -b_j, &\text{if}\ u_j= b_j,
\end{cases}
\label{eq:u-neq-b}
\end{equation}
we have
\begin{align}
0+[-1]^{b_j}u_j
&\leq\max\{1-b_j,-b_j\}
\notag
\\
&\leq 1-b_j
\end{align}
for all $j\in\K$.
This implies that $\vv=(0,u^k)$ satisfies (\ref{eq:poly-vb})
for all $j\in\K$.
From (\ref{eq:u-neq-b}), we have
\begin{align}
\sum_{j\in\K}[-1]^{b_j}u_j
&=
\sum_{\substack{
    j\in\K\\
    u_j\neq b_j
}}
\lrB{1-b_j}
-
\sum_{\substack{
    j\in\K\\
    u_j=b_j
}}
b_j
\notag
\\
&=
\sum_{\substack{
    j\in\K\\
    u_j\neq b_j
}}
1
-
\sum_{j\in\K}
b_j
\notag
\\
&\geq
1-\sum_{j\in\K}b_j,
\end{align} 
where the inequality comes from the fact that
there is $j'\in\K$ such that $u_{j'}\neq b_{j'}$
because $u^k\neq b^k$.
This implies that $\vv=(0,u^k)$ satisfies
(\ref{eq:poly-v0eq1}).
\end{proof}

\begin{lem}
\label{lem:c-vertex}
If $b^k\in\{0,1\}^k$,
then $(0,b^k)$ is not a member of $\cP(b^k)$.
If $b^k,u^k\in\{0,1\}^k$ satisfy $u^k\neq b^k$,
then $(1,u^k)$ is not a member of $\cP(b^k)$.
\end{lem}
\begin{proof}
First, we show the first statement.
From (\ref{eq:bk}), we have
\begin{align}
0+
\sum_{j\in\K}
[-1]^{b_j}
b_j
&=
-
\sum_{j\in\K}b_j
\notag
\\
&< 1-\sum_{j\in\K}b_j.
\end{align} 
This implies that
$\vv=(0,b^k)$ is not a member of $\cP(b^k)$
because it does not satisfy (\ref{eq:poly-v0eq1}).

Next, we show the second statement.
Let us assume that $b^k,u^k\in\{0,1\}^k$
satisfy  $u^k\neq b^k$.
Then there is  $j'\in\K$ such that 
$u_{j'}\neq b_{j'}$.
From (\ref{eq:u-neq-b}), we have
\begin{align}
1+[-1]^{b_{j'}}u_{j'}
&=
1+1-b_{j'}
\notag
\\
&>
1-b_{j'}.
\end{align}
Therefore, we have the fact that
$\vv=(1,u^k)$ is not a member of $\cP(b^k)$
because it does not satisfy (\ref{eq:poly-vb}).
\end{proof}

\begin{lem}
\label{lem:p-vertex}
The set of all integral vertexes of $\cP(b^k)$
is equal to $\cS(b^k)$.
\end{lem}
\begin{proof}
From Lemmas~\ref{lem:i-vertex}--\ref{lem:c-vertex},
we have the fact that all integral points of $\cP(b^k)$
are members of $\cS(b^k)$
and all members of $\cS(b^k)$ are integral points of $\cP(b^k)$.
In the following, we show that
all members of $\cS(b^k)$ are vertexes of $\cP(b^k)$.

We have
\begin{align*}
\{\vv\}=\bigcap_{j=0}^k\{\vv': v'_j=v_j\}
\end{align*}
for $\vv\equiv(v_0,v_1,\ldots,v_k)\in\cS(b^k)\subset\cP(b^k)$.
From Lemma~\ref{lem:valid} and the fact that $\vv\in\{0,1\}^{k+1}$,
we have the fact that
\begin{align*}
-v'_j\leq v_j,\ &\text{for $j\in\K$ s.t. $v_j=0$}
\\
v'_j\leq v_j,\ &\text{for $j\in\K$ s.t. $v_j=1$}
\end{align*}
are valid inequalities for $\cP(b^k)$.
This implies that
\[
\cP(b^k)\cap\{\vv': v'_j=v_j\}
\]
is a face of $\cP(b^k)$.
From Lemma~\ref{lem:face}, we have the fact that
\begin{align}
\bigcap_{j\in\K}
\lrB{\cP(b^k)\cap\{\vv': v'_j=v_j\}}
&=
\cP(b^k)\cap\bigcap_{j\in\K}\{\vv': v'_j=v_j\}
\notag
\\
&=
\cP(b^k)\cap\{\vv\}
\notag
\\
&=
\{\vv\}
\end{align} 
is a face of $\cP(b^k)$.
Since the dimension of $\{\vv\}$ is zero,
we have the fact that $\vv$ is a vertex of $\cP(b^k)$.
\end{proof}

\begin{lem}
\label{lem:v0frac}
For all $\vv=(v_0,v_1,\ldots,v_k)\in\cP(b^k)$
satisfying $0<v_0<1$,
$\vv$ is not a vertex of $\cP(b^k)$.
\end{lem}
\begin{proof}
Let $u^k\equiv(u^{(1)},\ldots,u^{(k)})$ be defined as
\[
u^{(j)}\equiv\frac{v_j-v_0b^{(j)}}{1-v_0}.
\]
Then $\vv$ can be expressed as
\begin{equation}
\label{eq:v0frac}
\vv = v_0(1,b^k)+[1-v_0](0,u^k).
\end{equation}
Since $\vv\neq(0,u^k)$ and $\vv\neq(1,b^k)$
for $0<v_0<1$,
we have the fact that
$\vv$ is not a vertex of $\cP(b^k)$ from Lemma~\ref{lem:vertex-iff}
by assuming that $\vv$, $(1,b^k)$, and $(0,u^k)$ are members
of $\cP(b^k)$.
Then it is enough to show that $(0,u^k)\in\cP(b^k)$
by assuming that $\vv\in\cP(b^k)$
because we have $(1,b^k)\in\cP(b^k)$ from Lemma~\ref{lem:vertex}.

In the following,
$v_0$ is replaced by $0$ and 
$v_j$ is replaced by $u^{(j)}$ in 
(\ref{eq:poly-v0})--(\ref{eq:poly-v0eq1})
when we state that $(0,u^k)$ satisfies inequalities
(\ref{eq:poly-v0})--(\ref{eq:poly-v0eq1}).
We show that
$(0,u^k)$ satisfies (\ref{eq:poly-v0})--(\ref{eq:poly-v0eq1})
by assuming $\vv\in\cP(b^k)$ and $0<v_0<1$.

First, it is clear that $(0,u^k)$ satisfies (\ref{eq:poly-v0}).
Next,
$(0,u^k)$ satisfies (\ref{eq:poly-v0b0}) and (\ref{eq:poly-v1b1})
because
\[
u^{(j)}=\frac{v_j-v_0b^{(j)}}{1-v_0}=\frac{v_j}{1-v_0}\geq 0
\]
for $j\in\K$ satisfying $b^{(j)}=0$ and
\[
u^{(j)}=\frac{v_j-v_0b^{(j)}}{1-v_0}=\frac{v_j-v_0}{1-v_0}\leq 1,
\]
for $j\in\K$ satisfying $b^{(j)}=1$,
where the inequalities come from the fact that
$v_0<1$ and $v_j$ satisfies (\ref{eq:poly-v0b0}) and
(\ref{eq:poly-v1b1}).
Next, $(0,u^k)$ satisfies (\ref{eq:poly-vb})
because
\begin{align}
0+[-1]^{b^{(j)}}u^{(j)}
&=
\frac{[-1]^{b^{(j)}}\lrB{v_j-v_0b^{(j)}}}{1-v_0}
\notag
\\
&\leq
\frac{1-b^{(j)}-v_0-v_0[-1]^{b^{(j)}}b^{(j)}}{1-v_0}
\notag
\\
&=
\frac{\lrB{1-b^{(j)}}\lrB{1-v_0}}{1-v_0}
\notag
\\
&=1-b^{(j)}
\end{align}
for all $j\in\K$,
where the inequality comes from the fact that
$v_0<1$ and $v_j$ satisfies (\ref{eq:poly-vb}),
the second equality comes from (\ref{eq:bk}).
Finally, we have the fact that 
$(0,u^k)$ satisfies (\ref{eq:poly-v0eq1})
because
\begin{align}
0+\sum_{j\in\K}[-1]^{b^{(j)}}u^{(j)}
&=
\sum_{j\in\K}
\frac{[-1]^{b^{(j)}}\lrB{v_j-v_0b^{(j)}}}
{1-v_0}
\notag
\\
&=
\sum_{j\in\K}\frac{[-1]^{b^{(j)}}v_j+v_0b^{(j)}}
{1-v_0}
\notag
\\
&\geq
\frac1{1-v_0}
\lrB{
  1-\sum_{j\in\K}b^{(j)}-v_0+v_0\sum_{j\in\K}b^{(j)}
}
\notag
\\
&=
1-\sum_{j\in\K}b^{(j)},
\end{align}
where
the second equality comes from (\ref{eq:bk})
and the inequality comes from
(\ref{eq:poly-v0eq1}) and the fact that $0<v_0<1$.
\end{proof}

Now we are in position to prove Lemma~\ref{lem:multi-binary}.
From Lemma~\ref{lem:p-vertex},
it is enough to show that there is no non-integral
vertex of $\cP(b^k)$.
Furthermore, from Lemma~\ref{lem:v0frac},
it is enough to show that
there is no non-integral vertex $\vv=(v_0,v_1,\ldots,v_k)$
of $\cP(b^k)$ by assuming $v_0\in\{0,1\}$.

Since $\cP(b^k)\subset\Real^{k+1}$,
the vertex of  $\cP(b^k)$ is determined by
$k+1$ equalities representing the face of $\cP(b^k)$
described by
\begin{align}
v_0&= 0
\label{eq:face-v0}
\\
v_j&= 0\quad\text{for $j\in\K$ s.t. $b_j=0$}
\label{eq:face-v0b0}
\\
v_j&=1\quad\text{for $j\in\K$ s.t. $b_j=1$}
\label{eq:face-v1b1}
\\
v_0+[-1]^{b^{(j)}}v_j&= 1-b^{(j)}\quad\text{for $j\in\K$}
\label{eq:face-vb}
\\
v_0
+\sum_{j\in\K}[-1]^{b^{(j)}}v_j
&=
1-\sum_{j\in\K}b^{(j)}.
\label{eq:face-v0eq1}
\end{align}
Since $v_0\in\{0,1\}$, we have the fact that
$v_j\in\{0,1\}$ if $v_j$ is determined by one of the equalities
(\ref{eq:face-v0b0})--(\ref{eq:face-vb}).

First, we consider the case where the vertex
$\vv=(v_0,v_1,\ldots,v_k)\in\cP(b^k)$ 
does not satisfy (\ref{eq:face-v0eq1}).
Then $\vv$ should be determined by equalities
(\ref{eq:face-v0b0})--(\ref{eq:face-vb})
and we have the fact that
$v_j\in\{0,1\}$ for all $j\in\K$.
This implies that $\vv$ is an integral point of $\cP(b^k)$.

Next, we consider the case where the vertex
$\vv=(v_0,v_1,\ldots,v_k)\in\cP(b^k)$ 
satisfies (\ref{eq:face-v0eq1}).
Since $k-1$ of $k$ variables $v_1,\ldots,v_k$ should be determined by
equalities (\ref{eq:face-v0b0})--(\ref{eq:face-vb}),
we have the fact that these variables are integers.
Then, from (\ref{eq:face-v0eq1}) and the fact that $v_0\in\{0,1\}$,
the remaining variable is also an integer.
This implies that $\vv$ is an integral point of $\cP(b^k)$.

Finally, from the above observations, we have the fact that
all vertexes $\vv=(v_0,v_1,\ldots,v_k)\in\cP(b^k)$ 
satisfying $v_0\in\{0,1\}$ are integral points.
This implies that there is no non-integral vertex $\vv=(v_0,v_1,\ldots,v_k)$
of $\cP(b^k)$ that satisfies $v_0\in\{0,1\}$.
\hfill\QED

\subsection{Method of Types}
\label{sec:type-theory}

We use the following lemmas for a set of typical sequences.

\begin{lem}[{\cite[Theorem 2.5]{UYE}\cite[Lemma 23]{HASH}}]
\label{lem:typical-trans}
If $\vv\in\T_{V,\gamma}$ and $\uu\in\T_{U|V,\gamma'}(\vv)$,
then $(\uu,\vv)\in\T_{UV,\gamma+\gamma'}$.
If $(\uu,\vv)\in\T_{UV,\gamma}$, then $\uu\in\T_{U,\gamma}$
and $\uu\in\T_{U|V,\gamma}(\vv)$.
\end{lem}

\begin{lem}[{\cite[Theorem 2.7]{UYE}\cite[Lemma 25]{HASH}}]
\label{lem:typical-aep}
Let $0<\gamma\leq 1/8$.
Then,
\begin{align*}
\begin{split}
  \left|
    \frac 1{n}\log\frac 1{\mu_{U}(\uu)} - H(U)
  \right|
  &\leq
  \zeta_{\U}(\gamma)
\end{split}
\end{align*}
for all $\uu\in\T_{U,\gamma}$,
and
\begin{align*}
\begin{split}
  \left|
    \frac 1{n}\log\frac 1{\mu_{U|V}(\uu|\vv)} - H(U|V)
  \right|
  &\leq
  \zeta_{\U|\V}(\gamma'|\gamma)
\end{split}
\end{align*}
for $\vv\in\T_{V,\gamma}$ and $\uu\in\T_{U|V,\gamma'}(\vv)$,
where $\zeta_{\U}(\gamma)$ and $\zeta_{\U|\V}(\gamma'|\gamma)$
are defined in (\ref{eq:zeta}) and (\ref{eq:zetac}), respectively.
\end{lem}

\begin{lem}[{\cite[Theorem 2.8]{UYE}\cite[Lemma 26]{HASH}}]
\label{lem:typical-prob}
For any $\gamma>0$, and $\vv\in\V^n$,
\begin{align*}
\mu_U([\T_{U,\gamma}]^c)
&\leq
2^{-n[\gamma-\lambda_{\U}]}
\\
\mu_{U|V}([\T_{U|V,\gamma}(\vv)]^c|\vv)
&\leq
2^{-n[\gamma-\lambda_{\U\V}]},
\end{align*}
where $\lambda_{\U}$ and $\lambda_{\U\V}$ are defined in (\ref{eq:lambda}).
\end{lem}

\begin{lem}[{\cite[Theorem 2.9]{UYE}\cite[Lemma 27]{HASH}}]
\label{lem:typical-number}
For any $\gamma>0$, $\gamma'>0$, and $\vv\in\T_{V,\gamma}$,
\begin{align*}
\left|
  \frac 1{n}\log |\T_{U,\gamma}| - H(U)
\right|
&\leq
\eta_{\U}(\gamma)
\\
\left|
  \frac 1{n}\log |\T_{U|V,\gamma'}(\vv)| - H(U|V)
\right|
&\leq
\eta_{\U|\V}(\gamma'|\gamma),
\end{align*}
where $\eta_{\U}(\gamma)$ and $\eta_{\U|\V}(\gamma'|\gamma)$
are defined in (\ref{eq:def-eta}) and (\ref{eq:def-etac}), respectively.
\end{lem}

\section*{Acknowledgements}
The authors wish to thank Prof.\ Uyematsu for introducing the problem of
the broadcast channel coding.
The authors also wish to thank Prof.\ Watanabe for introducing the
reference \cite{EK10}.
The authors also wish to thank anonymous reviewers for valuable comments.

\end{document}